\documentclass[10pt]{article}
\usepackage{amsthm,amsmath,graphicx,hyperref,lmodern,amsfonts,amssymb,amsfonts,amsthm}
\usepackage{upgreek,cancel,multirow,booktabs,mathtools,enumerate,accents,bm,multicol,subcaption}

\usepackage[svgnames]{xcolor} 

\newcommand\diff{\mathop{}\!\mathrm{d}}

\hypersetup{
  colorlinks,
  linkcolor={FireBrick},
  citecolor={MidnightBlue}
}
\usepackage[backend=biber,style=apa]{biblatex}
\newtheorem{theorem}{Theorem}[section]
\newtheorem{corollary}{Corollary}[theorem]
\newtheorem{lemma}[theorem]{Lemma}
\theoremstyle{definition}
\newtheorem{definition}{Definition}[section]

\newtheorem{remark}{Remark}

\usepackage{Sweave}
\begin{document}
\title{An unbiased non-parametric correlation estimator in the presence of ties}
\author{\href{mailto:ljrhurley@gmail.com}{Landon Hurley, PhD}}
\maketitle
% \twocolumn[
  % \begin{@twocolumnfalse}
    % \maketitle

\begin{abstract}
An inner-product Hilbert space formulation of the Kemeny distance is defined over the domain of all permutations with ties upon the extended real line, and results in an unbiased minimum variance (Gauss-Markov) correlation estimator upon a homogeneous i.i.d. sample.  In this work, we construct and prove the necessary requirements to extend this linear topology for both Spearman's \(\rho\) and Kendall's \(\tau_{b}\), showing both spaces to be both biased and inefficient upon practical data domains. A probability distribution is defined for the Kemeny \(\tau_{\kappa}\) estimator, and a Studentisation adjustment for finite samples is provided as well. This work allows for a general purpose linear model duality to be identified as a unique consistent solution to many biased and unbiased estimation scenarios.
\end{abstract}
  % \end{@twocolumnfalse}
% ]

% \tableofcontents
Correlation estimators have been a foundational part of the mathematics of linear model spaces. Building off of the constructive foundations of negative definite distance matrices upon complete metric space, guarantees concerning many highly desirable characteristics for correlation estimators, such as the central limit theorem, unbiasedness, finite sample generalisability, and guaranteed probability distributions have been provided. However while the Pearson correlation, which resolves around the well studied Euclidean distance or Frobenius norm function space is nearly ubiquitous, the Spearman and Kendall correlation estimators are substantially less well understood and available for researchers. This, we argue in this manuscript, is the direct consequence of an improper restriction of the domain to the symmetric group of order \(n\), (\(S_{n})\), representing the ordering of \(n\) data items on a vector of permutations with no ties for each variable.

This assumption is quite useful for data analysis beginning with continuous random variables and also in the asymptomatic analysis wrt \(n\), in which the rank-score bijection is ensured under the weak law of large numbers. Leveraging the birthday paradox, it is straightforward to conclude that the probability of a tie, the collision of two data elements with the same score on a random variable, but different values on the second,  tends quickly to 0 even for small samples. This result follows directly from the infinite cardinality of originating bivariate domain, which is equivalent to assuming continuous random variables. As has been already observed, the Wilcoxon Rank-Sum and the Mann-Whitney U tests are both equivalent to the respective Spearman and Kendall correlation estimate between a binomial grouping variable and a continuous endogenous response variable; effectively, these present a point-biserial correlation model framework, using different metric topologies (measures of distance function spaces). 
% This mathematical framework however allows for the equality between the standard Student t-test, the Pearson correlations, and proportionally the Analysis of Variance, to be replicated upon arbitrary ordered distributions in a non-parametric general linear model. At the same time though, we observe that the construction of Analysis of Covariance, or multiple regression, models is sorely lacking from a theoretical (and probabilistic) perspective, and is developed more fully in a subsequent publication. 
We present here an supersuming framework for the construction of the distribution-free estimators and test statistics, which are parametric (stochastically regular) in structure, in the sense of \textcite[Ch.~1]{hollander2014}.

% Conventional approaches, due to the explicit restriction to a domain without ties, inherently almost surely preclude the inclusion of multiple predictors and finite (discrete) response categories upon finite samples, as the finite set of choices precludes the traditional application of the Birthday paradox. This is due to the possibility of asymptotic unbiasednesss combined with a lack of a proof of stability upon finite samples. We explicitly demonstrate this problem in comparison to the existing Kendall's \(\tau_{b}\) and Spearman's \(\rho\) estimators, in comparison of a novel distribution free estimator and corresponding test statistic. In this manuscript, we have constructed an alternative pair of estimators for the Spearman \(\rho\) and a Kendall's \(\tau\) as specific subsets of a general linear model for non-parametric regression. We prove both to be finite sample minimum variance unbiased estimators which are also maximum likelihood estimators, and provide empirical demonstrations. Further, a framework for handling the standard t-test variations (e.g., Students' t-test and Welch's t-tests) for rank-based estimators is provided and the requisite characteristics of parametric estimators are expanded to the space of distribution free estimators with a regular (i.e., stochastic distribution following) probability structure. A likelihood function for these estimators is constructed from the Beta-binomial distribution, which is shown to allow both the construction of Wald tests and conventional estimators which are consistent in the first and second derivatives. 

These highly regular and unbiased minimum-variance estimator properties allow us to establish several highly noteworthy observations, including that the common non-parametric estimators are both biased and not a minimum variance estimator in the analysis of discrete data upon finite samples, compromising the generalisability of the partial Wald tests as put forth by \textcite{diaconis1977}. An explicit geometric duality, directly connected to Pearson's \(r\) is established between the novel estimators, and we show that under certain conditions, the non-linear bijection between  Kemeny's \(\tau_{\kappa}\) and Pearson's \(r\) is maintained. This duality is exceptionally powerful, allowing for non-positive definite (n.p.d) matrices and corresponding multivariate biased estimating equations upon the Pearson's \(\Sigma\) matrix to be uniquely solved according to empirical condition which satisfy the Gauss-Markov conditions and converge to an unbiased population solution as a Tikhinov regularised system of linear equations. The utility of these applications has been explored with positive results in applications to Linear Factor Analysis, solutions to Missingness Not at Random, and high-dimensional Kernel estimation framework solutions to general Karush-Kuhn-Tucker constrained optimisation problems.  

\section{A distribution free yet parametric general linear estimating function}

The theory of a standard general linear model requires an Hilbert, inner-product, metric space. These properties allow for convergence and regularity to be established for all observable elements upon finite independent and identical samples (i.e., stability). However, Spearman's \(\rho\) and Kendall's \(\tau\) estimators do not possess such a regular signed inner-product distance formulation, and only satisfy the necessary requirements of a Banach, or complete, metric space, which does asymptotically converge. Assume the existence of a sample matrix independently and identically sampled over the extended real line, \(\mathbf{X}^{n \times p} \in \overline{\mathbb{R}}\), indexed for \(i = 1,\ldots,n\) and \(j = 1,\ldots,p\) with restriction that \(n \gg p\). At this time, these assumptions are consistent standard practice, however we would note here that empirically investigation has demonstrated that a large ratio \(\tfrac{n}{p}\) is not necessary for the Kemeny metric to remain positive definite and thus unique. Allow lower case vectors, such as \(x_{i,j=1},y_{i,j=2} \in X^{n \times p=2}\) to denote univariate vectors of length \(n \times 1\) uniformally sampled upon the extended real line which are indexed upon matrix \(\mathbf{x}\). 

\subsection{First moment estimating equations upon the population}
We propose the following distance function, constructed as the Hadamard inner-product of two independently arising skew-symmetric matrices which represent the linear permutation space upon a basis \(\kappa: \overline{\mathbb{R}}^{n} \to n \times n\), which is indexed \(k,l = 1,\ldots,n\) for each data matrix column \(j \in \mathbf{X}_{n \times p}, j = 1,\dots,p\):

\begin{subequations}
\begin{equation}
\label{eq:kem_dist}
\rho_{\kappa}(x,y) = \frac{n^{2}-n}{2} + \sum_{k,l=1}^{n} \kappa_{kl}(x) \odot\kappa_{kl}^{\intercal}(y),~ k,l = 1,\ldots,n.
\end{equation}
\begin{minipage}{.45\linewidth}
\scriptsize
\begin{equation}
\noindent
\label{eq:kem_score}
\kappa_{kl}(x) = {
\begin{dcases}
\: \sqrt{.5} & \text{if } x_{k} > x_{l}\\
\: 0 & \text{if } x_{k} = x_{l},\\
\: -\sqrt{.5} & \text{if } x_{k} < x_{l}\\
\end{dcases}
}
\end{equation}
\end{minipage}
\begin{minipage}{.5\linewidth}
\begin{equation}
\scriptsize
\kappa_{kl}(y) = {
\begin{dcases}
\: \sqrt{.5} & \text{if } y_{k} > y_{l}\\
\: 0 & \text{if } y_{k} = y_{l},\\
\: -\sqrt{.5} & \text{if } y_{k} < y_{l}\\
\end{dcases}
}
\end{equation}
\end{minipage}
\end{subequations}

The \(\kappa\) function maps each extended real vector onto a skew-symmetric matrix of order \(n \times n\), \(\kappa: \overline{\mathbb{R}}^{n} \to (n \times n)\). Each entry in said \(\kappa\) matrix corresponding to vectors \(x,y\) are denoted by an entry in the \(k^{th}\) row and \(l^{th}\) column as element \(\kappa_{kl}(x_{i})\) and \(\kappa_{kl}(y_{i})\), respectively. Upon these two \(\kappa\) skew-symmetric matrices is then performed the Hadamard product (\(\odot\)), each representing the bivariate variable pair, whose inner-product is then summated over all \(n \times n\) elements. The distance calculated by the linear combination of the ordered vector space which results from the \(\kappa\) mapping, i.e., the permutation representation, of any two elements upon vector \(x_{i}\), allows for the linear representation of all permutations with ties, a complete space of cardinality \(\mathcal{M} = n^{n}-n\), rather than \(n!\), while excluding all constant \(n\) degenerate random variables. The matrix transpose is denoted by the superscript \((\cdot)^{\intercal}\).

% Note that the range of the distance function \(\rho_{\kappa}\) is upon \([0,1,\dots,n^{2}-n]\) and is therefore closed and totally bounded for all finite \(n\) variable vector lengths. Also note that under the \(\kappa\) function mapping, there exist \(n^{n}-n \subseteq \{\mathcal{M}\}_{m=1}^{n^{2}-n}\) distinct unique permutations for each bivariate set of \(n\) data items. The domain \(\mathcal{M}\) is constructed from the set of all permutations with repetitions, of which there are \(n\) data elements which may occur \(n\) times in the vector, while removing the degenerate set of \(n\) constant vectors from consideration. Thus, the population for each permutation upon a finite \(n\) Galois field is itself always finite. By observation, it is trivially concluded that the Kemeny metric is an even and therefore symmetric function and that the expectation of the distances in \(\mathcal{M}\) is \(\tfrac{n^{2}-n}{2}\), which is equal to 0 in equation~\ref{eq:kem_dist} if the leading term, the expectation upon the distance, is removed over all permutations in \(\mathcal{M}\).

We establish the existence of a neighbourhood around the expectation, which is equivalent to establishing the consistency and convergence of any random variables measured upon the abelian linear Kemeny function space. Let \(\rho_{\kappa} \in [0,n^{2}-n]\) as measured in equation~\ref{eq:kem_dist}. The largest possible Kemeny distance for any space in \(\mathcal{M}\) is obtained upon the Identity permutation \(I_{n \times 1} = 1,\ldots,n\) and its reverse \(I^{\prime}_{n \times 1} = n,\ldots,1\), whose distance is always \(n^{2}-n\) for all finite \(n\), and denotes a linear distance (as a metric space) over the extended real domain upon \(x,y\). \(E(\rho_{\kappa}(\mathcal{M}) = \tfrac{n^{2}-n}{2},\) and any affine linear transformation, including the subtraction of the expectation upon equation~\ref{eq:kem_dist}, is still a complete Hilbert space, by the definition of a metric space. Further, all metrisable spaces are recognised to be perfectly normal \(T_{6}\) Haussdorff spaces. By subtracting the set of all distances from about the expectation, the neighbourhood of the Kemeny metric is now signed about 0: \(U(\mathcal{M}) = [-\tfrac{n^{2}-n}{2},\tfrac{n^{2}-n}{2}]\). Said neighbourhood, which is almost surely finite as the mapping of \(n^{n} - n \mapsto [0,n^{2}-n]\), is an equivalent condition which establishes the consistency of the Kemeny distance. Therefore,

\begin{lemma}
% \label{lem:compact}
\label{lem:kem_bounded}
% By Definition~\ref{def:compact}
The Kemeny metric space is compact and totally bounded for any finite \(n\).% as it is complete and totally bounded.
\end{lemma}
\begin{proof}
The Kemeny metric is complete by Lemma~\ref{lem:kem_expansion}, and totally bounded, in that it possesses no point of finite $n$ reals, using finite positive real scalar  \(0 < a \in \mathbb{R}^{1 \times 1}<\infty^{+}\), which is outside the bounds of $0 \le \rho_{\kappa}(a\kappa^{*}(x_{A}),a\kappa^{*}(x_{B})) \le a(n^{2}-n), \forall\ 0 < a < \infty^{+}$ by Lemma~\ref{lem:complete_space}, where \(a\) is an arbitrary finite scalar. Therefore the Kemeny metric space is shown to be both compact and complete and consequently, separably dense as well as a \(T_{6}\) topological space.
\end{proof}

The Kemeny linear variance or concentration measure results from the summation of the \(n^{2}-n\) free parameters representing the finite support for all finite \(\kappa(x_{n \times 1})\): 
\begin{equation}
\label{eq:kem_variance}
\sigma_{\kappa}^{2}(x) = \tfrac{2}{n(n-1)}\Big(\sum_{k=1}^{n}\sum_{l=1}^{n} \kappa_{kl}(x)\kappa_{kl}^{\intercal}(x) \Big) \equiv \tfrac{2}{n(n-1)}\sum_{l,k=1}^{n}\kappa_{kl}^{2}(x).
\end{equation}
From this expression of the variance is defined a mapping of the extended real domain sub-space to a singular real which is positive if and only if \(x\) is non-degenerate. The scaling space reflects the squared (and therefore always non-negative) element composition, as guaranteed by Lemma~\ref{lem:kem_bounded}. When all squared elements in the skew-symmetric matrix are summed over, there exist \(n^{2}-n\) free elements, the diagonal \(n\) elements which strictly 0, producing a maximum order variance equals to 0.5 (proportional and equivalent to the scalar constant of \((\sqrt{.5})^{2}\), for each independent random variable under examination, which when computed upon the same variable twice is almost surely no greater than 1. The almost sure convergence and finiteness of the variance, as well as the mean, are guaranteed by Lemma~\ref{lem:kem_bounded}, a consequence of Kolmogorov's stronger order property. As the minimum of this same mapping is 0, which does not occur upon \(\mathcal{M}\), the support of equation~\ref{eq:kem_variance} is therefore \(\sigma_{\kappa}^{2} \in (0,1],\) for all possible non-constant vectors in the extended real line of any finite length \(n\). Note that, when considering the subset domain \(S_{n}\), \(\sigma^{2}_{\kappa}\) is for all element vectors a constant of 1. Therefore, the presence of ties upon the permutation space reduces the estimated variability of the sample, such that \(\sigma^{2}_{\kappa}(\mathcal{M} \setminus S_{n}) \in (0,1)\).

As each vector \(x\) is centred as the image of the \(\kappa\) matrix, one degree of freedom (the median) is lost in the estimation of the variance. This necessitates the introduction of a Bartlett correction, establishing the variance of the order statistics upon the Kemeny metric space in the presence of ties. By assuming independent and identical sampling, we obtain the following expression 
\begin{equation}
\label{eq:kem_population_variance}
\dot{\sigma}_{\kappa}^{2}(\mathcal{M}) = \frac{(n - 1)^2 (n + 4) (2 n - 1)}{18 n},
\end{equation}
representing the population variance of all Kemeny distances of a fixed length \(n\) between bivariate independent random variables, concordant with the empirical findings in Table~\ref{tab:1}. Note that this domain explicitly presumes that permutations are only sampled once, however which will be addressed later in Subsection~\ref{subsection:studentising}. 

These corrections for continuous random variables would naturally produce $n!$ such occurrences, which are indeed observed. We also note that under these conditions, limiting scenarios such as a Binomial random variable obtain identical variance when constructed upon either metric function space: this allows for a fundamental identity upon the orthonormal characterisation of finite data to be realised. At the moment however, consider that without the allowance for the measure of ties, the Kemeny variance function, and its equality to the Euclidean function space would be undefined. Therefore, as a viable tool for statistical learning, without the allowance for ties, multiple common scenarios become only asymptotically unbiased (and thus identified). 

We next proceed to prove the estimator constructed upon the distance function is unbiased:
\begin{lemma}
\label{lem:unbiased}
The Kemeny correlation (equation~\ref{eq:kem_cor}) is an unbiased estimator.
\end{lemma}
\begin{proof}
For the space of the Kemeny correlation, observe that upon the countable finite population \(\mathcal{M}\), that the expectation of the inner-product of the even function (as guaranteed by the symmetry of any metric space about \(U(\mathcal{M})\)) is of distance 0, and equivalently the correlation is 0. For \(n=2\), the set of \(\mathcal{M} = 2^{2} - 2 = 2\) permutations which possess a symmetric (by the even function nature of equation~\ref{eq:kem_dist}) neighbourhood of distances in \([-(\sqrt{0.5})^{2},+(\sqrt{0.5})^{2}]\). When the set of all elements \(m\) in \(\mathcal{M}\) is summated over, the population expectation \(\lim_{m \to\infty^{+}} E(U(\mathcal{M}) = 0\). Thus, the Kemeny estimator is asymptotically unbiased in expectation. By induction, observe the finite telescoping sequence of \(m_{i} \in \mathcal{M}, i = 1,\ldots,n^{n}-n\), from which follows both the finite expectation of 0 (by the closure under addition for the skew-symmetric matrix of equation~\ref{eq:kem_score} in the Kemeny metric space) and also the symmetry of the telescoping positive and negative distances, which arise by the even function nature of the \(\kappa\) function centred at the arbitrary point of origin (see Lemma~\ref{lem:even}). As all the absolute distances are never greater than \(\tfrac{n^{2}-n}{2}\), as the space is compact and totally bounded, regardless of the values which arise from the extended real line. For any finite \(n\), the sum and inner product of two random variables \(x,y\) of length \(n\) on the population \(\mathcal{M}\) indexed by \(m\mid{n}\) is \[\mu_{1} = \lim_{m\to \mathcal{M}}\tfrac{n^{2}-n}{2} + E_{m}(\tau_{x,y}) = \tfrac{n^{2}-n}{2} + \sum_{k,l = 1}^{n} \kappa_{kl}(x_{m})\odot\kappa_{kl}^{\intercal}(y_{m}) = \tfrac{n^{2}-n}{2} + 0,\] thereby completing the proof that the Kemeny correlation is an unbiased estimator of the median bivariate distance. 
\end{proof}

Transformation of the distance from the complete metric properties of the Hilbert space to a signed distance space is obtained by subtraction of the leading median distance, resulting in a finite subset \(m\) which in aggregate averages out to 0 in the neighbourhood of \(U(\mathcal{M})\) which is uniformly operated upon by the Glivenko-Cantelli theorem (Theorem~\ref{thm:gc}) for this linear function space. Further examination will demonstrate that the errors of approximation are themselves symmetrically distributed around the expectation (as the Kemeny metric is an even function; Lemma~\ref{lem:even}), and that by sub-additivity (Theorem~\ref{lem:hilbert}), the affine linear transformation by subtraction of the unbiased expectation does not bias the estimator, for all finite bivariate distances between random variables of length \(n\).

\subsubsection{Kemeny correlation \(\tau_{\kappa}\)}
Turning the distance function (equation~\ref{eq:kem_dist}) into an inner-product correlation coefficient is achieved in equation~\ref{eq:kem_cor},

\begin{equation}
\label{eq:kem_cor}
% \langle x,y\rangle_{\kappa} = 
\tau_{\kappa}(x,y) = -\tfrac{2}{n^{2}-n}\sum_{k=1}^{n}\sum_{l=1}^{n} \kappa(x)_{kl}\odot\kappa^{\intercal}(y)_{kl} = -\tfrac{2}{n^{2}-n} \Big(\rho_{\kappa}(x,y) - \tfrac{n^{2}-n}{2}\Big), \{x,y\} \in \overline{\mathbb{R}}^{n \times 1}.
\end{equation}
By taking the finite, compact, and totally bounded neighbourhood about 0, the distance \(\rho_{\kappa} \in [0,n^{2}-n]\), for arbitrary \(n\), and its support of the extended real domain of two extended real number lines of length \(n\), is transformed to a closed and compact neighbourhood centred at 0, \(U(\mathcal{M}) = [-\tfrac{n^{2}-n}{2},\dots,0,\dots,\tfrac{n^{2}-n}{2}]\). Multiplying by 2 rescales the extrema to span \(\pm (n^{2}-n)\) without changing the expectation, and the negation ensures that the further the realised distance from the origin of 0, which is always finite, the closer to a reverse permutation the random variable is. The extremal distance \(\tfrac{n^{2}-n}{n^{2}-n} = 1 \) under negation becomes \(-1\), supporting the interpretation that relative to the origin, the furthest distance is a vector for which all elements are completely reversed (termed the reverse Identity permutation vector, and noting that the origin is the Identity permutation vector). 

A scaled distance upon the neighbourhood at \(-(n^{2}-n)\), obtained by comparing the origin with itself after the median distance is subtracted,  may be rescaled to -1, and under negation is 1 (thereby denoting an identical ordering to the origin has distance 0, and has a cosine similarity of measure of \(\cos(0) = 1\)), as required for a correlational measure. Finally, all midpoint distances equal to the median distance are equal to 0, and any affine linear transformation upon them by the prescribed constants remains 0, denoting a normed correlation or distance measure which is fixed and uniquely defined at points on the finite field \(\{-1,0,1\}\) for all finite measures upon the neighbourhood about 0. This enables the utilisation of Brouwer's fixed-point theorem, as the isometric Euclidean distance and its cosine are non-linear at all other points.

\subsubsection{Theoretical properties of the estimator}
We now show that the Kemeny correlation estimator also satisfies the Gauss-Markov theorem, providing a best linear unbiased estimator which satisfies the Lehmann-Scheff\'{e} theorem, under the following definition:

\begin{definition}
\label{def:gauss_markov}
The necessary conditions to satisfy Gauss Markov theorem guarantee that the distance minimising function provides the best linear unbiased estimate (BLUE) possible point estimates upon a sample. These five Gauss Markov conditions are:
\begin{multicols}{2}
\begin{enumerate}
\footnotesize{
    \item{Linearity: estimated parameters must be linear.}
    \item{Variables arise i.i.d. by stochastic sampling from a common population.}
    \item{No variables are perfectly correlated.}
    \item{Exogeneity: the random variables are conditionally orthonormal.}
    \item{Homoscedasticity: the error of the variance is constant across the given population.}
}
\end{enumerate}
\end{multicols}
\end{definition}

\begin{theorem}
\label{thm:gauss-markov}
The Kemeny correlation is a Gauss-Markov estimator for any bivariate vector pair of length \(n\) which are independently sampled from a common population.
\end{theorem}
\begin{proof}
The linearity of the Kemeny function space follows by definition for any Hilbert norm space. The Kemeny distance and functions thereof (such as given in equation~\ref{eq:kem_cor}) are a proven in Hilbert space (see also Theorem~\ref{lem:hilbert}) by the existence a valid inner-product. Unbiasedness follows from Lemma~\ref{lem:unbiased}, and Gramian positive definiteness follows from either the satisfaction of the Mercer condition, or as the finite sum of the squared totally bounded positive variances of the Kemeny metric space (Lemma~\ref{lem:kem_bounded} enacted upon \(\mathcal{M}\); \cite{schoenberg1938}); both conditions are equivalent, and thus valid for the Kemeny measure space. Utilisation of Bochner's theorem to guarantee the Gramian nature is also valid, recognising that the Hilbert space is a continuous positive-definite function on a locally compact abelian group. Equivalently, the correlation between the two vectors must solely depend upon the distance between them, while remaining a.s. positive definite (i.e., a valid distance function) for any domain as defined by assumption. This property is valid for the Kemeny distance (equation~\ref{eq:kem_cor}) and the removal of the two collinear points guarantees that a positive definite and finite distance is always observed for all \(n < \infty^{+} \in \mathbb{N}^{+}\). Therefore, the correlation (and therefore covariance) matrix is always positive definite, for any population which is neither collinear or degenerate. 

Condition 3 is satisfied by axiomatic assumption, wherein the unique points \(I,I^{\prime} \setminus \mathcal{M} \equiv Q\) (uniqueness guaranteed by the Riesz representation theorem for a Hilbert space, upon a topological function space of monotone affine-linear invariance). For arbitrary random sampling upon the Kemeny neighbourhood for finite \(n\) though, there are exist only 2 \(\kappa\) mappings \(\mathcal{M} \setminus Q\)  which are collinear for any pair of independent random variables. The probability of observation (see also Theorem~\ref{thm:clt_kem}) quickly then tends to 0 in the limit wrt \(n\): \[ \Pr(\mathcal{M} \cap {Q}) = \lim_{n\to\infty^{+}} \frac{|Q| = 2}{n^{n} - n} = 0.\] Exogeneity in turn follows from the Riesz representation theorem for any Hilbert space, as independent random sampling follows under axiomatic assumption of the theorem domain such that \(\lim_{n \to \infty^{+}} \tfrac{2}{n^{n}-n}\). The homoscedasticity assertion follows by definition of the Kemeny variance \(0<\sigma^{2}_{\kappa}<\infty^{+}\) whereupon only linearly comparable scores may be ordered, as there may only exist one set of permutations upon any common population function (Cumulative Distribution Function). This excludes therefore the possibility of unique measurement upon periodic functions, which is valid by axiomatic assumption, else there must exist realisations upon the extended real line which may not be validly assessed by equation~\ref{eq:kem_score}. This completes the proof that the Kemeny correlation satisfies all necessary requirements of a Gauss-Markov estimator.
\end{proof}

The non-negative definiteness of each variable is given by the square of a skew-symmetric matrix (whose transpose provides a negation which is then squared), with 0 obtained only upon \(n\) degenerate sample distributions, and is therefore otherwise excluded from consideration upon \(\mathcal{M}\). The sum of any such bivariate sequence of \(n^{2}-n\) such numbers, a given consequence for the skew-symmetric matrix, which imposes a diagonal of \(n\) 0's, multiplied by \(\sqrt{0.5}^{2}\) must almost surely be contained in the interval \((0,\dots,n^{2}-n]\), which by the properties of a metric space, is almost surely positive and also trivially confirmed to be finite for all finite \(n\).
\subsection{Probability regularity of the Kemeny distance and all affine-linear transformations}
Assume without loss of generality that the extended real vector space of \(p\) variates are expressible upon a positive definite variance-covariance matrix \(\Xi_{p\times p}\). With the fulfilment of the definition of the observation of a positive finite variance for any non-constant random sequence, i.e., a variable, along with the compact support of the neighbourhood about 0 for the signed distance measure on \(U(\mathcal{M})\) (equation~\ref{eq:kem_variance}), the distribution of the distances upon \(U(\mathcal{M})\) satisfies the necessary conditions to be sub-Gaussian \parencite[Ch.~1]{buldygin2000}.

% We observe that there must therefore also exist a sequence of moments (as trivially guaranteed for any finite Banach norm space and also any compact and totally bounded domain, Lemma~\ref{lem:kem_bounded}). As has been proven, the first expected central moment is equal to a centred 0, as calculated in equation~\ref{eq:kem_dist}, and the second moment is always finite, for all finite \(n\). The square of such any finite skew-symmetric \(n \times n\) matrix is always negative semi-definite, and therefore any sum of squares upon a totally bounded compact space is itself almost surely positive finite for any bivariate random variable pair, and, by induction, the sum of any sequence of such sets of finite positive real scalars.

The expansion to the necessity of the minimum necessary moment sequence (i.e., \(n\)is not a sufficient population parameter) is now discussed. A sub-Gaussian variable \(\xi\) is said to be strictly sub-Gaussian if and only if
\begin{definition}
\label{def:stict_sg}
\begin{equation}
\label{eq:strict}
E(e^{\lambda\xi}) \le e^{\frac{\lambda^{2}\sigma^{2}_{\xi}}{2}}, \forall~ \lambda \in \overline{\mathbb{R}} \equiv \|\xi\|_{\kappa} \le \sigma_{\xi} = \|\xi\|\ell_{2}(\Omega).
\end{equation}
\end{definition}

\begin{theorem}
\label{thm:strict_subgauss}
The distribution of the Kemeny distance is strictly sub-Gaussian (Definition~\ref{def:stict_sg}) for any finite sample \(n < \infty^{+}\), and is therefore stable.
\end{theorem}

\begin{proof}
The distribution of the Kemeny distance is, for all finite \(n\), defined upon the neighbourhood \(U(\mathcal{M})\), and is therefore almost surely finite and measurable (the existence of a perfectly normal Hilbert space ensures the existence of a Borel \(\sigma\)-measure). The finite definiteness of all possibly observable variances has also already been established (Lemma~\ref{lem:kem_bounded}) upon \(\mathcal{M}\). A sub-Gaussian variable is defined as the conjunction of a finite variance and compact totally bounded support, and therefore the Kemeny distance is sub-Gaussian for finite \(n\) \parencite[Ch.~1]{buldygin2000}. For either the bivariate or univariate case then, we observe a sub-Gaussian distribution: in the univariate scenario, the spectrum of the distance measure is \((0,b], b \in (0,\ldots,\mathbb{N}^{+},\infty^{+})\), with the trivial restriction that \(n>0\) to allow \(b \ne 0\), and for which the expectation under affine transformation of 0, satisfies Condition~\ref{eq:strict}. For the bivariate scenario, the spectrum of the vector inner-product for the \(\kappa\) skew-symmetric matrices \(\{\kappa_{x},\kappa_{y}\}\) is also finite and totally bounded, and is also therefore strictly sub-Gaussian for any finite \(n\), as their distance is almost surely less than or equal to the Euclidean distance \(\ell_{2}\)-norm uniformly across \(U(\mathcal{M})\). Thus, both the marginal distribution and bivariate distributions are strictly sub-Gaussian for all finite \(n\).

For the contrapositive condition, consider if for finite \(n\) the distribution of the population of Kemeny distances were Gaussian. If this were so, then there would exist measurable events which occur with non-zero density upon the Gaussian pdf but which are not observable for any finite \(n\): 
\begin{equation}
\lim_{n\to\infty^{+}} \Pr(x \setminus U(\mathcal{M}) \in \pm \tfrac{n^{2}-n}{2}) =  \int_{\infty^{-}}^{-\tfrac{n^{2}-n}{2}} \frac {1}{\sigma {\sqrt {2\pi }}} e^{-{\frac {1}{2}}\left({\frac {x-\mu }{\sigma }}\right)^{2}}\diff{x} - \int_{-\tfrac{n^{2}-n}{2}}^{\infty^{+}} \frac {1}{\sigma {\sqrt {2\pi }}} e^{-{\frac {1}{2}}\left({\frac {x-\mu }{\sigma }}\right)^{2}}\diff{x} \ge 0, ~x \in \mathbb{R}.
\end{equation}
As such, there must exist for all populations of permutations observable distances upon the Gaussian probability distribution which do not occur upon the Kemeny distance support, and thus equality is only obtained in the limit, wherein equality to 0 holds. Thus for the distribution is not Gaussian over the domain \(x^{n \times 1},y^{n \times 1} \in \overline{\mathbb{R}}, n < \infty^{+}\).
\end{proof}

\begin{corollary}
If then the Kemeny distance is almost surely strictly sub-Gaussian, we observe that the distribution is centred at 0, and possesses symmetric tails of density which is less than or equal to that of a standard normal distribution. Therefore it immediately follows that four moments are sufficient to characterise the probability distribution upon any population of size \(\mathcal{M}\).
\end{corollary}
\begin{proof}
It is self-evident that any power of an expectation of 0 is also 0, and therefore that all higher order odd-moments are equal to 0 for the Kemeny metric (thus also establishing an even symmetric function in distribution over all \(m\)). The second central moment is defined in equation~\ref{eq:kem_variance} and is therefore orthonormal of the expectation of 0; by the symmetry of the distance (Lemma~\ref{lem:even}) the skewness is 0, and therefore the final free moment to examine is the excess kurtosis \(\mu_{4}\) upon a finite sample. The Kemeny distribution however is strictly sub-Gaussian and therefore must posses negative excess kurtosis for any finite \(n\), and thus cannot be normally distributed. This unique probability distribution is therefore symmetric and unbiased, has a spread and measure of all scores which is almost surely positive and finite, possesses no skewness (by the even function property), and a finite negative excess kurtosis which tends to 0 asymptotically from below \parencite[Ch.~1]{buldygin2000}. This paradoxically contradicting the extrapolated asserted conclusion that the asymptotic performance holds for finite samples (i.e., is stable), as the distribution for the finite sample (strictly sub-Gaussian) distances cannot be normally distributed in the presence of ties upon finite samples.
\end{proof}
% In turn, asymptotic unbiasedness is an explicit condition which denotes that said distribution of non-parametric estimators is not stable: therefore, the finite sample performance must be explicitly biased in scenarios which require stronger conditions than which may be observed upon the stochastic data generating procedure. This provides an entropic framework of understanding which allows for Studentising to be performed upon ordered data vectors: this is necessary, as the common population variance otherwise defined preserves conditions such as mutually exclusive sampling or even stronger conditions such as the occurrence of no ties. 

% Such conditions are exceptionally difficult to maintain in the presence of multidimensional covariate spaces, and therefore it becomes much more effective to construct \(t\)-distributions from which to analyse Frequentist decision frameworks. %We also note that while not developed in this manuscript, the explicit construction of a order-based likelihood framework which is functionally orthonormal yet also dual to the conventional variate valued function space, must allow for greater information which does not bias, but instead results in more information being used to characterise a system of linear equations.  

% This is concordant with the claim that the maximum variance will be observed when all elements $n$ upon $x$ are observed to occur uniquely, such that no duplication occurs (and therefore there will exist \(n\) unique order statistics for each random variable). 
An interesting corollary is noted, wherein the Kemeny distance and correlation functions are shown to be related to the Kendall $\tau$ distance and affine linear functions thereof including correlations. The Kendall \(\tau\) function space is therefore a strict subset of the larger, complete, Kemeny space.

\begin{corollary}
 \label{cor:density}
 Assume that for $n>0$, the density of the permutation spaces for the respective measures are $\mathcal{M} = n^{n} - n$ for the Kemeny metric space, and $\mathcal{M}^{\prime} = {n}!$ for the Kendall metric space, where $\mathcal{M}^{\prime} \subset \mathcal{M}$. It would then follow that the number $\mathcal{M}$ must be well represented for the Stirling approximation of factorials \[{\sqrt {2\pi n}}\ \left(\frac {n}{e}\right)^{n}e^{\frac {1}{12n+1}} < {n}! < {\sqrt {2\pi n}}\ \left({\frac {n}{e}}\right)^{n}e^{\frac{1}{12n}},\] or else \(\mathcal{M}^{\prime} \subset \mathcal{M}\).\end{corollary}
 \begin{proof}
 Assume $\mathcal{M} \subseteq \mathcal{M}^{\prime}$, denoting that the $\tau$-distance space is of greater or equal density to the Kemeny $\rho_{\kappa}$-space. If this were so, it would therefore follow that \(2n \le \log\big(2\pi n\big)\) (see Lemma~\ref{lem:density} for completed expansion), which is strictly false by contradiction under all conditions, as there is no $n\in \mathbb{N}^{+}$ which satisfies this strict inequality. It follows then that $\mathcal{M}^{\prime} \subset \mathcal{M}$, and therefore that all measurements upon the Kendall distance are a strict subset within the Kemeny distance, $\tau(x,y) \subset \rho_{\kappa}(x,y)$. As the permutations with ties must always strictly subsume the set of all permutations without ties, the set difference must always be positive \(\lim_{n\to\infty^{+}} n^{n}-n - n! > 0,\) a relational which holds uniformly for all \(n>2\), and thereby ensures that the density of the set-space of all permutations with ties is almost surely greater.
 \end{proof}

 \begin{corollary}
The linear proportionality of the two distance functions is obtained upon the subset $x$ for which $\sigma_{\kappa}^{2} = \frac{n^{2}-n}{2}$, the set of $n!$ such permutations without ties. For this set, observe that $\tau(x,y) \in [0,\frac{n^{2}-n}{4}]$, whereas $\rho_{\kappa}(x,y) \in [0,\frac{n^{2}-n}{2}]$, and that for each permutation, $\tau(x,y) = \frac{1}{2}\rho_{\kappa}(x,y)$. \end{corollary}
 \begin{proof}
As a direct consequence of the unit distance proportional to $a \in \overline{\mathbb{R}}^{+}$\footnote{Note that the explicit identification of scalar constant \(a = |\sqrt{.5}|\) is unnecessary.} for the $\tau$ function occurs such that for $ x = \{a,b\} \overset{\tau}{\to} y = \{b,a\} = \tau(x,y) = 1a.$ However, for $x = \{a,b\} \overset{\tau}{\to} y = \{b,a\} = \rho_{\kappa}(x,y) = 2a,$. This demonstrates that a distance of 1 upon the Kendall metric is half of the Kemeny metric and that when proportionally rescaled, the minimum non-existent distance of 1 of \(S_{n}\) measured by the Kemeny metric reflects the measurement wherein a tie occurs, but which are likewise not observed. As no ties occur upon the $\tau$ distance domain (\(S_{n}\)), the minimum distance of a tie relative to a total ordering is non-measurable, consistent with Corollary~\ref{cor:density}, and therefore an affine-linear bijection exists $\tau(x,y) = 1a \propto \rho_{\kappa} = 2a \ \therefore 2\times \tau(x,y) = \rho_{\kappa}(x,y), \forall\ (x,y) \in U(\mathcal{M}),$ assuming both random variables are non-degenerate and of finite common length \(n\). Therefore, every $\tau$-distance is validly and uniquely defined upon the $\rho_{\kappa}$-distance, but not vice-versa, as there exist no Borel \(\sigma\)-mappings for a measure space in which a tie is observed to occur upon Kendall's \(\tau\). Thus, the population of Kendall's \(\tau\) is a strict subset of the Kemeny \(U(\mathcal{M})\).
 \end{proof}
\begin{remark}
Note that this is equivalently more easily established by the defined support of \(\tau = \tfrac{n^{2}-n}{2}\) and \(\rho_{\kappa} = n^{2}-n\), which share common extrema (and origin) of a uniquely defined distance 0 for identical random variable pairs, and bivariate permutations which are proportionate by scalar \(2\) for every other distance. Induction from \(n=2; \mathcal{M} \cap S_{n}\) upon \(n\) immediately establishes the proportionality of the common distances. 
\end{remark}

For this linear function space then, we establish the asymptotic normality of the distance and correlation estimator:
\begin{lemma}~\label{lem:kem_asym_normal}
The bivariate strictly sub-Gaussian Kemeny metric space is asymptotically \(wrt n\) normally distributed.
\end{lemma}
\begin{proof}
Consider the kurtosis of the probability distribution of the Kemeny distance function for the Beta-Binomial distribution. As a function of \(n\), the random variable functional representing the Kemeny distance between bivariate independent random variables, is trivially observed to be a monotonically increasing function. By definition of strict sub-Gaussianity the excess kurtosis is always negative, for all finite \(n\) with asymptotic normality (and thus obtains an excess kurtosis of 0 on in the limit wrt \(n\)). By the existence of all odd-moments equal to 0 and a positive variance, the monotonically convergent kurtosis tends to 3 from below in the limit of \(n\). This maintains the strict sub-Gaussianity over all finite sample sizes, which becomes degenerate for the asymptotic limit of \(n\to\infty^{+}\) and which is unbiased for all \(n\) over the population of \(\mathcal{M}\). 

Proof by contradiction assures the asymptotic normality of the set \(U(\mathcal{M})\). Upon the population limit \(n\to\infty^{+}\), the cardinality of the set diverges from the real line and there is no collection of finite distances or cardinalities, as \(U^{*}(\mathcal{M}) \supset U(\mathcal{M})~\forall n\). Mapping said ratio of the elements onto the finite support of the neighbourhood satisfies the upper-limit of Kolmogorov's stronger order property, yet also produces a degenerate bound upon an infinite set. Upon said set, the distances upon an infinite set of elements may no longer be uniquely ordered upon a finite distance relative to the unique origin \(I_{n}\), which is realised as equivalent to a probability measure of size \(0 = (U^{*}(\mathcal{M}))^{-1}\), and therefore cannot be stably strictly sub-Gaussian. Therefore, upon the strict bijection between rank and score measure distances, the Gaussian distribution is sufficient and uniquely capable of denoting said linear function space. Upon the population, by the weak law of large numbers a perfectly Gaussian random variable is observed upon which the linear bijection of rank and score is observed to hold almost surely. Therefore, the divergent Kemeny distance is bijectively collinear with the Gaussian probability density function and its integral, once centred and scaled, and thus both measure spaces upon a common distribution is uniquely known by the Euclidean distance.
\end{proof}

While the neighbourhood of the Kemeny metric space is strictly sub-Gaussian and only asymptotically normally distributed, the operation function space upon the square vector matrix characterisation (that of the mapping \(\kappa\)) presents implementation difficulties upon finite samples (as we may typically assume, but rarely verify, that an asymptotically normal distribution is also normally distributed upon finite samples). This relationship is developed later, in Subsection~\ref{subsection:studentising}. 
% The utilisation of a probabilistic optimisation function as an explicit transition function is clearly a highly desirable optimisation function space, but one for which we are unable to find existing developments upon. While a vector probability estimate is directly estimable to provide maximum likelihood procedures, the convolution of a vectorised \(n \times 1\) probability distribution with an  \(n \times n\) data matrix is uncertain, as the transition across permutation matrices, i.e., the explicit linear nature of the gradient upon a vector matrix space, is currently uncertain. We hypothesise that this is a fundamental grounding underlying the Markov Chain relationship to continuous variable spaces, and using the duality between the Kemeny metric and Euclidean metric function spaces, would be able to uniquely identify a stopping point for finite sequence transitions.

\begin{lemma}
\label{lem:lsn}
The Kemeny metric satisfies the strong law of large numbers for any identically and independently distributed pair of independent random variables as a linear distance function.
\end{lemma}
\begin{proof}
By Markov's inequality the existence of a finite expectation satisfies the strong law of large numbers, which is guaranteed for any finite \(n\) upon the Kemeny metric, as by the totally bounded nature of the Kemeny distance function, follows \(E(I\{\kappa(x_{i}) \le \kappa(x)\}) < \infty\) for all finite \(n\) then. This condition is equivalent to establishing that \(\sup{\sqrt{(\pm \frac{n^{2}-n}{2})^{2}}} = \frac{n^{2}-n}{2} \equiv \sup{\rho_{\kappa}(x,y)}, \,\forall~  n \in \mathbb{N}^{+}\), is easily verified. A second condition trivially holds that for any finite vector sequence \(\{x_{i}\}_{i=1}^{n}\), a linear ordering may be obtained, such as by using the image of the \(\kappa\) mapping from any random variable independently and identically sampled upon the extended real line. Thus, the strong law of large numbers is observed to hold for any finite sample upon \(U(\mathcal{M})\) for which \(\kappa(\cdot)\) may be validly measured.
\end{proof}
By the strong law of large numbers then, we ensure that a linear function converges asymptotically to the true distribution, and one which does so stably for all finite \(n\), by the strict sub-Gaussian nature of its distribution (Theorem~\ref{thm:strict_subgauss}). Further, for each \(n\) we may partition the observation of an independently random vector \(x\) within \(\mathcal{M}\) as a consequence of the orthonormal relationship conditionally observed in isolation (i.e., assuming i.i.d.). Therefore, an unbiased estimator of the CDF of \(F_{\kappa}\) arises, which by Lemma~\ref{lem:lower_equality} and Lemma~\ref{lem:upper_equality} are shown to converge to the true distribution by the strong law of large numbers, regardless of the non-linear score relationship by the Total Variational Distance upon the \(\sigma_{\kappa}\)-algebra of the Kemeny metric space. The distribution however is not self-evident, however by empirical moment matching, we obtained the distribution of the Kemeny distance to be Beta-Binomial in distribution for all finite \(n\). The compound nature of the probability distribution is a conventional problem, however as the support is almost surely compact and totally bounded with regular probability, the Binomial nature of the distribution may be removed for any sample by multiplication of \(U(\mathcal{M}) \cdot \sup U(\mathcal{M})\), which cumulatively normalises the support from 0 to 1.

\subsection{Sufficiency of the Kemeny parameter space}
Given a set $x$ of independently distributed data conditioned upon an unknown parameter vector $\theta$, a sufficient statistic is a function $T(x)$ whose value contains all information necessary to compute any estimate of the parameter. By the factorisation theorem, for a sufficient statistic $T(x)$, the probability density may be then expressed as $f_{x} = h(x)g(\theta,T(x))$; the unknown parameter vector $\theta$ is composed of functions upon the original data space which are sufficient, in that for the conditional probability distribution of the data $x$, given statistic $t = T(x)$, $(x \perp \theta \mid t)$. A statistic is sufficient if its mutual information with $\theta$ is strictly equal to the information between $x$ and $\theta$ (by the monotonicity of the squared function, all non-equal elements are strictly less). Thus, it will been shown that the sufficient statistic $\mu_{\kappa} = \sum_{j=1}^{n}\sum_{i=1}^{n} \kappa_{ij}(x)$ is sufficient for the central tendency of any distribution which is independently distributed upon the reals. 

By additionally satisfying $\mu_{\kappa}^{1}(x) = 0$ for any $x \in \overline{\mathbb{R}}^{n \times 1}$ is found a sufficient and unbiased estimator of location for which $E(\kappa_{i}(x)) = E(\kappa(x)) = 0\; \forall\, \theta \implies \kappa(x) = 0\; \forall\, \theta,$ almost everywhere (Lemma~\ref{lem:kem_asym_normal}). Therefore, the Kemeny metric satisfies the necessary conditions of the Lehmann-Scheff\'{e} theorem for completeness, and therefore it is also the the best unbiased estimator function, as was already shown by Theorem~\ref{thm:gauss-markov}. Sufficiency therefore follows by the asymptotic normality upon the Kemeny measure space, and for all \(n\) by the strict sub-Gaussianity. %The variance \(\sigma_{\kappa}^{2}\) is a sufficient statistic, with the fourth central moment, \[\lim_{n\to\infty^{+}} \gamma^{4}_{\kappa}(x_{n})= (\rho_{\kappa}(x,y)-\mu_{\kappa}^{1})^{4} = f(n) \to 3.\] We conclude that for any finite \(n\) of finite size the probability distribution of the Kemeny distance sampled without replacement upon \(\mathcal{M}\) follows a Beta distribution with shape parameters \(\alpha_{1} = \alpha_{2} = n\), the sample size, and which is discretised onto the observed support \([0,\ldots,\tfrac{n^{2}-n}{2}\) with the application of the Binomial function of size \(n^{2}-n\) with adjacent step probabilities drawn from a Beta distribution with the determined shape parameters.

\begin{theorem}~\label{thm:clt_kem}
Let \(\{x_{1},\ldots,x_{n}\}\) be a random sample of size $n$ drawn from a distribution of expected value given by $\mu$ and finite variance given by $\sigma^{2}$. By the law of large numbers upon the Kemeny metric function space (Lemma~\ref{lem:lsn}), the sample converges almost surely to the expected value $\mu$ for the limit wrt $n$, as proven here by the Lindeberg-L\'{e}vy Central Limit Theorem:

Suppose $\{x_{1},\ldots ,x_{n}\}$ is a sequence of i.i.d. random variables with $E[x_{i}]=\mu$ and $\operatorname{Var}[x_{i}]=\sigma^{2}<\infty$. Then as $n$ approaches infinity, the random variables \({\sqrt {n}}({\bar {x}}_{n}-\mu )\) converge in distribution to a normal distribution \(\mathcal{N}(0,\sigma ^{2})\): \[\sqrt {n}\left({\bar {x}}_{n}-\mu \right)\ \xrightarrow {d} \ {\mathcal {N}}\left(0,\sigma ^{2}\right).\]
\end{theorem}
\begin{proof}
The necessary conditions to validly apply the (generalised) central limit theorem to the strictly sub-Gaussian Kemeny metric are threefold. First, substitute $\lim_{n\to\infty^{+}} \frac{1}{n} \equiv |\mathcal{M}|^{-1} = 0,$ a necessary conversion for an atomless probability measure for any finite \(n\). For finite $n$ then, the population of all permutation vectors with repetition upon $n$ independently sampled observations is not atomless, but of population cardinality $|\mathcal{M}| = n^{n}-n$ elements, explicitly rejection all \(n\) degenerate yet observable constant vectors of length \(n\). This produces a complete measure space of probability \(\tfrac{1}{|\mathcal{M}|}\), for which the probability of non-observability tends to 0, and the probability of observation of a point which is non-measurable over \(\mathcal{M}\) tends to 0, almost surely and thus the entire non-degenerate extended real domain is linearly measurable: \[1 -\lim_{m\to\infty^{+}} \Pr(\mathcal{M} \setminus\kappa(x)) = 1.\] 

This is confirmed by the Markov inequality, which shows the absence of any non-measurable finite vectors upon the mapping function \(\kappa\), ignoring degenerate distributions, and which tends to 0 for any sequence of increasing length \(n\). 
Second, let $\mu^{1}_{\kappa}$ exist, as proven in both Lemma~\ref{lem:unbiased} and Lemma~\ref{lem:markov} for the Kemeny metric. Finally, let the variance be finite for any population, which has also proven valid upon the Kemeny metric space (Lemmas~\ref{lem:kem_bounded} and \ref{lem:chebyshev}). By these three conditions it therefore follows that the Central Limit Theorem is a valid construction upon the Kemeny metric space, ensuring a.s. linear convergence in distribution for any sequence of $p$ unknown yet identified linear estimands, which are strictly sub-Gaussian distributed variables for finite \(n\).  
\end{proof}

An equivalent conclusion for the Central Limit Theorem using Lyapunov's condition (Corollary~\ref{cor:lyapunov}) and Kolmogorov's law (Corollary~\ref{cor:kolmogorov}) upon the Kemeny metric space are given in the Appendix.

\begin{theorem}\label{thm:gc}
The Kemeny metric function over \(\mathcal{M}\) satisfies the Glivenko-Cantelli theorem: Let \(\{x_{i}\}_{i=1}^{\mathcal{M}}\) be an independently distributed uniform sequence of random variables with distribution function \(F \in \overline{\mathbb{R}}\). Then \[\sup_{x\in\overline{\mathbb{R}}} | \hat{F}_{m}(x) - F(x)| \to 0, a.s.\]
\end{theorem}

\begin{proof}
For any \(\epsilon>0\), holds \[\lim_{m\to\infty}\sup_{x\in\overline{\mathbb{R}}} |\hat{F}_{m}(x) - F(x)| \le \epsilon, a.s.\]  By Lemma~\ref{lem:partition} exists a partition index \(j\) for which \(t_{j} \le x < t_{j+1}\), satisficing:
\footnotesize
\begin{align*}
\hat{F}_{m}(t_{j}) \le \hat{F}_{m}(x) \le \hat{F}_{m}(t^{-}_{j+1}) \land F(t_{j}) \le F(x) \le F(t^{-}_{j+1}),\\
\implies \hat{F}_{m}(t_{j}) - F(t^{-}_{j+1}) \le \hat{F}_{m}(x) - F(x) \le \hat{F}_{m}(t^{-}_{j+1}) - F(t_{j}) \equiv \\
\hat{F}_{m}(t_{j}) - F(t_{j}) + F(t_{j}) - F(t^{-}_{j+1}) \le \hat{F}_{m}(x) - F(x) \le  \hat{F}(t^{-}_{j+1}) - F(t^{-}_{j+1}) + F(t^{-}_{j+1}) - F(t_{j})\\
\therefore \hat{F}_{m}(t_{j}) - F(t_{j}) - \frac{\epsilon}{2} \le \hat{F}_{m}(x) - F(x) \le \hat{F}_{m}(t^{-}_{j+1}) - F(t^{-}_{j+1}) + \frac{\epsilon}{2},
\end{align*}\normalsize
which tends to equality at 0 by the strong law of large numbers. Thus, the rank ordering of any extended real distribution satisfies the Glivenko-Cantelli theorem upon the Kemeny metric for any finite and therefore countable sample.
\end{proof}

\subsection{Haar measure}
\label{subsec:haar}
The existence of Haar measures allows us to define admissible procedures such that optimal invariant decision criteria may be established. More pragmatically, they allow us to take the probability structure guaranteed by the existence of a Banach norm-space, such as invariant transformation of optimisation functionals which are admissible for both Frequentist and Bayesian inference. In particular, Haar measures allow the construction of prior probabilities and conditional inference in statistics (thereby enabling extensions to function spaces including multiple regression). Thus, we establish here a number of mathematical primitives upon the Kemeny metric space, including the existence and validity of Radon derivatives and integrals over the invariant volume of the sample, rather than the population. Both properties are foundational within statistical analysis, and are now shown to be validly enacted upon the Kemeny function space as well. Further, we note that by the properties of the Haar measure, as \(\mathcal{M}\) and functions thereof are discrete, the countably additive left-invariant measure may be wholly defined on all subsets of the domain, by the axiom of choice.

A function as a Haar measure is defined as a unique countably additive, non-trivial measure \(\mu\) on the Borel subsets of \(G\) satisfying the following properties:
\begin{definition}~\label{def:haar}
\begin{multicols}{2}
\begin{enumerate}
    \item{The measure \(\mu\) is left-translation-invariant: \(\mu (gS)=\mu (S)\) for every \(g\in G\) and all Borel sets \(S\subseteq G\).}
    \item{The measure \(\mu\) is finite on every compact set: \(\mu (K)<\infty^{+}\) for all compact \(K\subseteq G\)}
    \item{The measure \(\mu\) is outer regular on Borel sets \(S\subseteq G\): \[\mu (S)=\inf\{\mu (U):S\subseteq U\}.\]}
    \item{The measure \(\mu\) is inner regular on open sets for compact \(K\) \(U\subseteq G:\)\[\mu (U)=\sup\{\mu (K):K\subseteq U\}.\]}
\end{enumerate}
\end{multicols}
A measure on \(G\) which satisfies these conditions is called a left Haar measure, and is a sufficient and necessary condition to establish right Haar measure existence and proportionality, and therefore equivalence.
\end{definition}

\begin{lemma}~\label{lem:radon}
The Kemeny metric space satisfies Definition~\ref{def:radon} of a Radon measure space, thereby proving for the existence of a Radon derivative.
\begin{definition}~\label{def:radon}
If \(x\) is a Hausdorff topological space, then a Radon measure on \(x\) is a Borel measure \(m\) on \(x\) such that \(m\) is locally finite and inner regular on all Borel subsets.
\end{definition}
\end{lemma}
\begin{proof}
The Kemeny metric is a Hilbert metric space (by equation~\ref{eq:kem_dist}), and thus is a \(T_{6}\), or perfectly normal space Hausdorff topological vector space. In consequence, \(\kappa(x), x \in \overline{\mathbb{R}}^{n \times 1}\) must be both locally finite and inner regular. By the Riesz representation theorem, all metric spaces are inner regular on open sets \(K\). By Lemma~\ref{lem:kem_bounded}, the metric space is always locally finite. Thus the finite Borel measure upon the Kemeny metric is tight (in the sense of \cite{bogachev2007} Theorem 7.1.7), and there exists a Radon derivative upon the Kemeny measure.
\end{proof}

\begin{lemma}
\label{lem:haar}
The Kemeny metric space satisfies all properties of Definition~\ref{def:haar} and is therefore a Haar measure space over all random independent variables sampled independently and identically from the extended real line.
\end{lemma}
\begin{proof}
The Kemeny metric and its Borel \(\sigma\)-algebra are closed under addition and multiplication, and therefore are left-translation-invariant. The total boundedness of Lemma~\ref{lem:kem_bounded} guarantees the measure \(\kappa(K)\) is finite for all \(K \subseteq \overline{\mathbb{R}}\). The inner-regularity is proven in Lemma~\ref{lem:radon} and both outer regularity and completeness follow as a Hilbert space.
\end{proof}

% The strict sub-Gaussian nature of the Kemeny distribution guarantees that four moments are sufficient, and that the distribution is symmetric (resulting in all odd-moments being centred at 0 w.l.g) and unbiased in expectation. Here, we approach this probability function from a discrete perspectives, using a %continuous distribution (the Beta-Binomial distribution) and second a
% Beta-Binomial distribution, which unlike a normal distribution, is compact and totally bounded finite \(n\), and therefore strictly sub-Gaussian.

\begin{table*}[!ht]
\tiny
\caption{Distributional characterisation of the Kemeny distance function, for finite \(n\), permuted exhaustively for \(n \le 8\), and sampled with 3,294,172 examples for all larger \(n\).}
\label{tab:1}
\centering
% \footnotesize
\begin{tabular}{lccc | lccc }
\toprule
 \(n\) & \(\mu_{1}\) & \(\sigma_{\kappa}\) & Excess \(\mu_{4}\) &  \(n\) & \(\mu_{1}\)  & \(\sigma_{\kappa}\) & Excess \(\mu_{4}\) \\% &  \(n\) & \(\mu_{1}\) & \(\sigma_{\kappa}\) & Excess \(\mu_{4}\)\\
\midrule
 2 & 0.000 & 0.707 & -1.875 &    25 & 0.002 & 42.647 & -0.091\\ %&   80 & -0.066 &  240.736 & -0.023\\
   3 & 0.000 & 1.610 & -1.171 &    26 & 0.019 & 45.183 & -0.083\\ % &   85 & 0.0643 & 263.268 & -0.0274\\
   4 & 0.000 & 2.646 & -0.747 &    27 & -0.066 & 50.477 & -0.080\\ %& 92 & 0.144 & 296.597 & -0.023\\
   5 & 0.000 & 3.795 & -0.548 &    28 & 0.040 & 53.177 & -0.076\\ %& 96 & 0.155 & 315.766 & -0.0276\\
   6 & 0.000 & 5.046 & -0.432 &    29 & 0.040 & 55.900 & -0.075\\ %& 98 & & & \\
   7 & 0.000 & 6.392 & -0.356 &    30 & 0.046 & 58.674 & -0.072\\ %& 100 & -0.038 & 335.703 & -0.024 \\
   8 & 0.000 & 7.826 & -0.302 &    31 & 0.011 & 61.506 & -0.069\\ %& 103 & & & \\
   9 & 0.000 & 9.345 & -0.259 &    32 & 0.043 & 64.418 & -0.066\\ %& 105 & 0.0280 & 360.907 & -0.022 \\
   10 & 0.006 & 10.939 & -0.230 &    33 & 0.014 & 67.287 & -0.061\\ %& 108& & & \\
   11 & 0.009 & 12.622 & -0.212 &    34 & 0.082 & 70.272 & -0.062\\ %& 112 & & & \\
   12 & -0.007 & 14.352 & -0.191 &    35 & -0.016 & 73.262 & -0.065\\ %& 115 & & & \\
   13 & -0.017 & 16.168 & -0.173 &     36 & -0.036 & 73.262 & -0.057\\ %& 123 & & & \\
   14 & 0.006 & 18.064 & -0.161 &    37 & -0.005 & 76.255 & -0.063\\ %& 125 & 0.021 & 468.456 & -0.013 \\
   15 & -0.010 & 19.996 & -0.148 &    38 & 0.010 & 79.419 & -0.060 \\%& 126 & & & \\
   16 & -0.025 & 22.005 & -0.141 &    40 & -0.035 & 85.764 & -0.058 \\%& 128 & & & \\
   17 & 0.001 & 24.066 & -0.131 &    45 & -0.052 & 102.107 & -0.052 \\%& 135 & & & \\
   18 & 0.010 & 26.216 & -0.122 &    50 & 0.035 & 119.342 & -0.043\\ %& 138 & & & \\
   19 & -0.021 & 28.386 & -0.117 &    55 & 0.075 & 137.645 & -0.039\\ %& 143 & & & \\
   20 & 0.025 & 30.645 & -0.105 &    60 & 0.057 & 156.656 & -0.036 \\%& 147 & & & \\
   21 & 0.006 & 32.942 & -0.105 &    62 & -0.100 & 164.530 & -0.039 \\%& 155 & & & \\
   22 & 0.008 & 35.272 & -0.096 &  64 & -0.020 & 172.447 & -0.039 \\%& 175 & & & \\
   23 & 0.031 & 37.694 & -0.096 &    68 & -0.032 & 188.808 & -0.030 \\%& 200 & & & \\
   24 & -0.012 & 40.155 & -0.095 &    75 & 0.042 & 218.527 & -0.026 \\%& 225 & -0.056 & 1127.979 & -0.013 \\
\bottomrule
\end{tabular}
\end{table*}

As the distribution function of the Kemeny distances upon \(U(\mathcal{M})\) is known to be the Beta-Binomial, the excess kurtosis is an explicit function of \(n\), and therefore does not require separate estimation. Instead, the non-normality of the distribution is resolved by the convergence of the symmetric distribution to the standard normal distribution, for the asymptotic limit wrt \(n\), by the exponential characterisation of the Beta distribution (see also Lemma~\ref{lem:kem_asym_normal}). An alternative proof may be presented by noting that the probability distribution is Beta distributed, which is exponential, and therefore must almost surely converge to the normal distribution under i.i.d. sampling conditions. The desirability of focusing upon only the Beta distribution (and safely ignoring the Binomial component, which serves only to discretise the continuous probability manifold, a necessary characteristic of a Hilbert space), is treated upon in Subsection~\ref{subsec:mle} to construct the MLE.

% We note that the discrete nature of the Kemeny distance support is an explicit result of the Binomial distribution of size \(n^{2}-n\), and thus the probability regularity is defined upon the Beta distribution, which is hereafter concluded to be a continuous and sufficient statistic.

\subsubsection{Cochran's theorem upon arbitrary uniformly sampled variables}
If we accept that a Beta-Binomial distribution is necessary to restrict the support to the interval of the \(\kappa\) function and thus the Kemeny metric function (Lemma~\ref{lem:kem_bounded}), then the univariate distribution of the variance must also be a central \(\chi^{2}_{\nu=1}\) distribution, using the modified Bessel function of the first kind \(I_{v}\) of the form:
\begin{equation}
f(x\mid \sigma_{\kappa},n) = \frac{1}{2}\bigg(\frac{x}{\sigma^{2}_{\kappa}}\bigg)^{\frac{M-1}{2}}\exp\Bigg(-\frac{x+\sigma^{2}_{\kappa}}{2}\Bigg)I_{v}(\sigma_{\kappa}\sqrt{x}),~ x = \kappa(x)
\end{equation}
\begin{lemma}
\label{lem:cochran}
The Kemeny estimator satisfies Cochran's theorem under uniform sampling.
\end{lemma}
\begin{proof}
To prove that the Kemeny estimator satisfies Cochran's theorem under uniform sampling, consider the set of \(\kappa^{2}\) values for any domain space \(\mathcal{M}\), wherein the variance is strictly positive and less than one, and for which the first-order expectation is always 0 for all \(m\) (by Lemma~\ref{lem:unbiased}). By the Chernoff bound then, the known tails of the CDF may be obtained at the truncated points \(a = \frac{n-n^{2}}{2},b = \frac{n^{2}-n}{2}\), which defines the support of \(U(\mathcal{M})\) for all finite \(n\). As the \(\kappa\) function is the linear basis of a complete metric Hilbert space, it is closed under both addition and multiplication, and therefore allows for the sum of \(\chi^{2}\) variables to also be distributed as such. This is accepted due to the linear addition upon the domain \(x,y\) surely possessing a corresponding co-image mapping by equation~\ref{eq:kem_score}, which is invariant under linear and monotone translation; corresponding conditions for this extension were already proven in \textcite{semrl1996}, and therefore we conclude that Cochran's theorem holds for randomly sampled variables of length \(n\) upon the Kemeny measure space, subject to the assertion that all scores are linearly orderable (such as arising from the exponential family of distributions, or any discretisation thereupon). It then directly follows that the \(\chi^{2}_{1}\) distribution holds, by the normality of the approximately linear projective function Hilbert space allows us to accept Chernoff's bound for the discrete uniform distribution without issue, purely as a function of the already proven generalised central limit theorem for strictly sub-Gaussian random variables, assuming only uniform sampling independence upon observed linearly orderable distributions.
\end{proof}

In consequence of the valid use of the Kemeny metric upon independent random variables which may be ordered upon a common population, follows Lemma~\ref{lem:cochran}. By this result, we obtain the following signed Wald test statistic with a corresponding normed Beta-binomial distribution, which is asymptotically normal for \(n\to\infty^{+}\) and unbiased for all finite  \(\mathcal{M}\) (per both the exponential nature of the Beta distribution, and the asymptotic normality per Lemma~\ref{lem:kem_asym_normal}):
\begin{equation}
\label{eq:z_kemeny}
z_{\phi}(x,y) = \frac{-\rho_{\kappa}(x,y)}{\sqrt{\dot{\sigma}^{2}_{\kappa}}},
\end{equation}
denoting the ratio of the centred Kemeny distance \(\rho_{\kappa}, E(\rho_{\kappa}) = 0,\) which is normed by the standard deviation over \(\mathcal{M}\mid{n}\). It is immediately self-evident that \((z_{\phi})^{2}\sim\chi^{2}_{df = 1}\). As there is only one degree of freedom for constant known function of the sample size (which produces the variance per equation~\ref{eq:kem_population_variance}), this distribution is a unit normal z-statistic, which estimates the Kemeny distance between two vectors upon a finite sample drawn from a common population (note that the standard deviation, rather than the variance, is utilised, due to the numerator being a signed measure, rather than a variance of the distances from the unique origin). More generally, while not explored in this manuscript, Cochran's theorem (Lemma~\ref{lem:cochran}) allows for the general linear model theory to be applied (e.g., the use of multiple regression upon distribution free estimators which are unbiased in expectation upon finite samples), by establishing unique identification of conditional multivariate distributions, as well as valid approximation using unique Bayesian methodologies.

Upon the expanded domain \(\mathcal{M} \setminus S_{n}>0\), Kendall's \(\tau_{b}\) is biased for all finite samples, and therefore not a minimum variance estimator. This is confirmed by empirical simulation in Tables~\ref{tab:non_par-location}, \ref{tab:non_par-effect}, and \ref{tab:3}, along with the theoretical proof that non-measurable events occur when \(x^{n \times 1},y^{n \times 1} \not{\in} S_{n}\). These tables exhibit the expected theoretical minimum variance properties of the Kemeny correlation (and resulting Wald statistics) for discrete ordinal data, evidencing that the established alternative estimators are produce both overly dispersed (and positively biased) empirical results upon the Kendall \(\tau\) estimator. Thus, it must necessarily follow that Type I errors are inflated for given \(\alpha\)-thresholds, when applied to non-continuous data (such as bivariate ordered finite response sets or counts). Common examples include both the ordinal scale items common to behavioural assessments, and graded response models, such as item reviews, which demonstrate that the linear distance is non-constant upon the Frobenius norm or Kendall norms, but is upon the Kemeny norm. An empirical preference is consequently demonstrated for the Kemeny correlation estimator over the Kendall correlation estimator for all scenarios, as the distribution is otherwise identical to the Kendall's \(\tau\) in limited continuous and asymptotic scenarios, whereupon ties do not occur, but is otherwise both more accurate and more efficient. All referenced experimental simulations presented were conducted  using the asymptotic approximation of the variance for the Kendall \(\tau_{b}\) correlation estimation and test statistic \(z_{b}\).%:

\section{Estimating equations for the non-parametric Kemeny correlation}
The previous derivations provided analytic solutions and proofs for the Kemeny \(\rho_{\kappa}\) distance functional topological framework which provides in turn the normed correlation. In turn is provided the commiserate Wald test statistics, and proofs of the unbiasedness and minimum variance characteristics of the Kemeny estimation for a number of standard bivariate analytical scenarios (i.e., any common bivariate population whose elements are orderable upon a CDF, upon the extended real line). The Beta-Binomial distribution of the random variables' pairwise distances, used to construct the Wald statistics, is also amenable to maximum likelihood estimation. This allows us to establish that equivalent, in principle, to the Gaussian distribution, the Kemeny metric is linear maximum likelihood estimator (MLE) for the non-parametric bivariate correlations, well as an unbiased minimum variance estimator across a much broader set of cumulative distribution functions which are independently sampled. An important characteristic is noted here upon the duality between the Kemeny correlation and the Spearman correlation, which can in turn then be extended to Pearson's \(r\) under suitable parametric restrictions. Finally, we provide a framework for Studentising the Kemeny correlation space, allowing for the construction of general purpose Student's t-tests estimators, Welch's t-tests estimators, and paired sample t-tests.

\subsection{Method of Moments Estimators}
Given the known population characteristics, a method of moments estimator of the shape parameters \(\alpha_{1}\) and \(\alpha_{2}\) follows for the Beta distribution, which arises by the division of the Beta-Binomial population of distances by \(N = n^{2} - n\), thereby norming, for any sample of any finite size, the distance to be a compact and totally bounded elements on the real interval \(\tfrac{\rho_{\kappa}}{N} \in [0,1]\). 
\begin{remark}
Note that in examining the \(\mathrm{B}\), or Beta, distribution, it has been empirically observed that the central limit theorem (in the form of the strong, rather than weak, law of large numbers) is only valid for \(n \ge 9\) as has been previously noted. Using the product (or ratio) of the variance of the symmetric \(\mathrm{B}\) distribution to the support for each \(n \in 1,\ldots,\mathbb{N}\), we note that the variance is only greater than 1 if and only if \(n\ge 9\):
\begin{equation}
\begin{aligned}
 E[x]_{\kappa}                 & = \frac{\alpha_{1}}{\alpha_{1} + \alpha_{2}} = \frac{n}{2n} = \frac{1}{2}\mid \alpha_{1} = \alpha_{2} = n\\ 
\operatorname{var}[x]_{\kappa} & = \frac{\alpha_{1}\alpha_{2}}{(\alpha_{1}+\alpha_{2})^2(\alpha_{1}+\alpha_{2}+1)} = \frac{1}{8n + 1}\\
\end{aligned}
\end{equation}
such that \((\tfrac{n^{2} - n}{8n+1})> 1 \iff n> 9.\) This defines the standard pseudo-median construction uniquely under less restrictive conditions \parencite[p.~39]{hollander2014}, and thereby allows for a minimum variance unbiased estimator of said location upon the extended real domain. This in turn suggests that the variance must be smaller than the support of the measure space in order to allow the central limit theorem to apply to a unimodal distribution, and in turn allowing the densest partition of the support to be uniquely defined upon the monotonically convergent space. The implications of this biased or constrained optimisation problem (uniquely solved via KKT estimation) is addressed and developed in later work. 

Also note that the Beta distribution of the Kemeny estimator is asymptotically normal for \(\lim_{n\to\infty^{+}}\) (Theorem~\ref{thm:strict_subgauss}), while otherwise remaining a closed and totally bounded even function space for all finite samples, supporting this selection of probability distribution. This is also confirmed by the exponential nature of the Beta distribution, from which follows the asymptotic normality of the estimating equations. We also note here that examining the ratio \(\tfrac{\mathcal{M}}{n!}\) we observe a markedly non-exponential rate of change at approximately \(n \ge 15.\) In particular, note that upon this ratio, the density at the standard CLT threshold of \(n \ge 20\) is observed to be, on the \(\log(\log(n))\) scale, approximately linear.   
\end{remark}

% \begin{equation}
% \begin{aligned}
% \hat{\alpha_{1}} = \bar{x} \{\frac{\bar{x}(1 - \bar{x})}{s^{2} - 1} \}\\
% \hat{\omega} = (1 - \bar{x}) \{ [ \frac{\bar{x}(1 - \bar{x})}{s^{2}}] - 1 \}
% \end{aligned}
% \end{equation}
% where
% \begin{subequations}
% \begin{equation}
% \bar{XY} = \frac{1}{n} \sum_{i=1}^n \underline{XY}_{i}
% \end{equation}
% \begin{equation}
% \label{eq:mmue}
% s^{2} = \frac{1}{n-1} \sum_{i=1}^n (\underline{XY}_{i} - \bar{XY})^2,
% \end{equation}
% \end{subequations}
% with equation~\ref{eq:mmue} presenting the unbiased Bessel corrected sample variance estimator.
Solving for free parameters \(\alpha_{1},\alpha_{2}\) in the estimating equations for the Beta distribution, for which both \(n,\rho_{\kappa}\), the sample size and bivariate Kemeny distance, are known for the random data vectors (variables) \((x,y)\) for the identification restrictions that \(\alpha_{1},\alpha_{2},n,\rho_{\kappa} >0\). These are in line with pre-existing assumptions, and does not restrict the solution space beyond the assumptions already imposed to prove the Gauss-Markov conditions.

% First, consider that the variance of \(\tau_{\kappa}(x,y) = \sigma_{\kappa}(x)\sigma_{\kappa}(y) = \sigma(xy).\) As the product of the two marginal variances is equal to the variance of the uniformly sampled without replacement permutation projection space observed for fixed \(n\), then equation~\ref{eq:kem_variance} is a sufficient statistic for the variance, which is unbiased upon the population of permutations \(m = 1,\ldots,n^{n}-n \subseteq \mathcal{M}\), assuming xy is the Kemeny product distance of \(n \times 1) vectors \(x,y\), and allowing as previously shown that \(E(\langle xy\rangle)_{i=1}^{\mathcal{M}} = 0\), and which represents the expectation of scores as the median function of \(n\):
% % \begin{aligned}
% % E(\sigma^{2}_{\kappa}) & = E(\tfrac{1}{n-1})\sum_{i=1}^{n^{n}-n} (\langle{xy}\rangle_{i} - \bar{xy})^{2} = E(\tfrac{n^{2}-n}{2} \pm n^{2}-n\\
%                        % & = \frac{1}{n-1} E(\sum_{i=1}^{n}(xy)^{2}_{i} - 2\sum_{i=1}^{n}x_{i}\bar{xy}\bar{y}_{i} + \sum_{i=1}^{n} \bar{x}^{2}\bar{y}^{2})
% \frac {\partial \ln {\mathcal {L}}(\rho \mid X,\alpha_{1} = n,\alpha_{2} = n)}{\partial \rho }}=\sum _{i=1}^{N}\ln X_{i}-n{\frac {\partial \ln \mathrm {B} (n ,n )}{\partial \rho }}=0
% 
% \end{aligned}

Using the Beta component of the Beta-binomial distribution, which controls the shape of the probability distribution upon which the discrete support \([0,n^{2}-n]\) arises, we observe that the population probability distribution of the errors is as follows:
\begin{equation} 
\begin{aligned}
 f(x\mid \alpha_{1} = \alpha_{2} = n ) & = \mathrm {constant} \cdot x^{n-1}(1-x)^{n-1}\\
                    & = \frac {x^{n-1}(1-x)^{n-1}}{\int_{0}^{1}u^{n-1}(1-u)^{n-1}\,\diff{u}}\\ 
                    & = \frac{\Gamma(n+n)}{\Gamma(n)\Gamma(n)}\,x^{n-1}(1-x)^{n-1}\\
                    & = \frac{1}{\mathrm{B}_{I}(n,n)}x^{n-1}(1-x)^{n-1},
\end{aligned}
\end{equation}
where \(\Gamma\) is the gamma function, and the incomplete Beta function \(\mathrm{B}_{I}\) normalises the function to maintain the summation to 1 required for the Haar measure upon the finite population \(\mathcal{M}\). 

Accepting that \(\mathrm{B}(\hat{\alpha_{1}},\hat{\alpha_{2}}\mid x,y)\) serves as the likelihood function for the probability distribution for the sample Kemeny distance, whose support is validly discretised by the Binomial distribution of size \(N\). It then follows that the MLE is equivalent to the analytic solution presented in equation~\ref{eq:kem_dist}, if and only if in the limit over \(m \to \mathcal{M} = n^{n}-n\), we observe that the concave probability distribution (which we have already noted to be exponential, and therefore numerically, as well as analytically, solvable) must be maximised at a common point estimate, represented as a linear function of \(\hat{\alpha_{1}},\hat{\alpha_{2}}\); acceptance of the Frobenius norm representation of the probability measure directly follows by the isometric equivalence of all metric spaces upon the domain of the population. If we only accept the probability structure of the moments to be Beta-Binomial distributed for the distance between the two random variable vectors which are Beta distributed, then

\begin{subequations}
\begin{equation}
\label{eq:moment1}
\frac{\alpha_{1}}{\alpha_{1} + \alpha_{2}} = \frac{\rho_{\kappa}}{(n^2 - n)}
\end{equation}
\begin{equation}
\label{eq:moment2}
\frac{\alpha_{1} \alpha_{2}}{(\alpha_{1} + \alpha_{2})^2 (\alpha_{1} + \alpha_{2} + 1)} = \frac{(n - 1)^2 (n + 4) (2 n - 1)}{18 n\cdot (n^2 - n)}
\end{equation}
\begin{equation}
\label{eq:moment1_shape1}
\hat{\alpha}_{1} = \tfrac{\rho_{\kappa} (18 \rho_{\kappa}^{2} + 2n^{5} + n^{4} - 19 n^{3} - 18 n^{2} \rho_{\kappa} + 31 n^{2} + 18n \rho_{\kappa} - 19n + 4)}{n(n - 1)^{4}(2 n^{2} + 7n - 4)}
\end{equation}
\begin{equation}
\label{eq:moment2_shape2}
\hat{\alpha}_{2} = \tfrac{(n^{2} - n -\rho_{\kappa}) (18 \rho_{\kappa}^{2} + 2n^{5} + n^{4} - 19 n^{3} - 18 n^{2} \rho_{\kappa} + 31 n^{2} + 18n \rho_{\kappa} - 19 n + 4)}{n(n - 1)^{4}(2 n^{2} + 7 n - 4)},
\end{equation}
\end{subequations}
from which are estimated \((\hat{\alpha}_{1},\hat{\alpha}_{2})\) for the joint distribution of the co-product \(Z = xy\), and which is identified for any \(n,\rho_{\kappa}\; x,y \in \overline{\mathbb{R}}^{+}\). Then, solving for the distance using the sample estimated shape parameters (equations~\ref{eq:moment1_shape1} and \ref{eq:moment2_shape2}) from which then results the correlation estimate (via equation~\ref{eq:kem_cor}), we observe a functional and convergent estimator of the probability distribution of the Kemeny distance between two population vectors, conditional upon the observed data for a given finite sample size \(N = n^{2}-n\). 

From the characterisation of the sample dependent distribution parameters \(\hat{\alpha}_{1}\) and \(\hat{\alpha}_{2}\), we observe that we also possess a deterministic representation of the variance for Beta distribution, as given in the left-hand side of equation~\ref{eq:moment2}. 

Of course, as a direct consequence of the uniqueness for the compact and totally bounded support Kemeny distance space, the variance for either extremum is likewise uniquely determined and thus uniquely degenerate, as there exist only \(|m| \in \mathcal{M}  = 2 \mid \rho_{\kappa} \in \{0,n^{2}-n\}\). Thus, the empirical estimation of the uncertainty of said estimator is minimised in the population, and can also be appropriately studentised for finite samples (see Section~\ref{subsection:studentising} for details), with the provided upper bound for addressing the limiting case.

\subsubsection{Marginal Method of Moments Estimator}
With the bivariate joint distribution for the Kemeny distance between two vectors of length \(n\) which are uniformly sampled from a common population, we know proceed to likewise characterise  the marginal distributions, again using the method of moments. We accept that the distance between a given vector and itself is always uniquely defined at \(0\) for any Hilbert space, including the Kemeny distance function space. Therefore, continuing the use of the Beta distribution for the probabilities of the Kemeny distance function, we estimate \(\alpha_{\bullet}\), the shape parameter of the marginal distribution \(X\). By the symmetry of the uniform distribution, there is only one unknown, identified as a function of the variance with equation~\ref{eq:kem_variance} scaled by \(n^{2}-n\) for a sample drawn independently upon a common population: 
% \[\alpha_{1} = \frac{1}{8} (\frac{1}{\sigma^{2}} - 4),\]
\begin{subequations}
% \begin{align*}
\begin{equation*}
\frac{\alpha_{1}^{2}(n^{2}-n)}{4\alpha_{1}^{2}(2\alpha_{1}+1)} = \dot{\sigma}_{\kappa}^{2}
\end{equation*}
\begin{equation}
\alpha_{1} = \frac{(n^{2} - n - 4 \sigma^{2}_{\kappa})}{8 \dot{\sigma}^{2}_{\kappa}},
\end{equation}
\end{subequations}
for the Beta distribution, which, when scaled by \(2^{t}(n^{2} -n)\), provides the necessary characteristics of the Beta-Binomial distribution which is observed as the Kemeny distance (notably that of the compact and totally bounded finite support for the extended real domain). Using the location of the order-centred statistics (corresponding to the point \(X = \{x_{(1)},\dots,x_{(i)}\}\) upon which the sum of the ordered symmetric elements of \(\kappa_{kl}\) satisfy \(\sum_{k=1}^{n}\kappa_{kl}(x) = 0\)), is obtained several immediately recognisable characterisations. 

First, the order statistic vector, which is uniquely proportionately defined, is quadratic in nature and thus obtained at the corresponding point for which \(\kappa^{2}\) distance is minimised as well, at the expectation. For didacticism, we observe that this score equation (solved as the derivative of the \(2^{t}\) power of all elements, \(\kappa^{2}(\cdot)\)) is thus sufficient to identify the first moment of any marginal random variable which is uniformly drawn upon a non-degenerate population. Further extending the \(\kappa\) function as the basis linear operator of the function space identified by the Kemeny distance, the \(t^{th}\) for \(t = 1,\dots,\mathbb{N}^{+},\) order-moment for random variable \(X\) is equal to \(\frac{1}{\sqrt[t]{2^{t}(n^{2}-n})}\sum_{kl}^{n}\kappa^{t}_{kl}(x)\), as a function of \(\alpha_{1} \in [0,\infty^{+}).\) Therefore, the marginal distribution of \(X\) is just-identified by \(\hat{\alpha}_{1}\) and the observed vector length (i.e., the sample size), independent of any affine linear transformation, which allows for the expression of the variance. %The characterisation of a non-negative scored Beta distribution is achieved by the addition of the maximum score upon the \(n \times 1\) order-statistic vector to the centred random variable which results from the summation of the skew-symmetric \(\kappa_{kl}(X_{i,j})\) matrix over all \(k\) rows with the corresponding \(\sup{X_{(i)}}\), the largest order statistic which arises upon the distribution with freely estimated, and therefore unknown, variance.

\subsubsection{Numerical Example}
A trivial numerical example is provided from the Fisher Iris data, for which \(n = 150\). The pairwise distances and the analytically solved parameters are provided in Table~\ref{tab:example_moments}, and it may be immediately observed that the product of the support scale and the expectation of the Beta function is the Kemeny distance: 
\begin{align*}
\mu (n^{2}-n) = \tfrac{\alpha_{1}}{\alpha_{1}+\alpha_{2}}(n^{2}-n) = \rho_{\kappa} \mid \rho_{\kappa} = 13145, \alpha_{1} = 0.5797333, \alpha_{2} = 0.4059677\\
\tfrac{0.5797333}{0.5797333+ 0.4059677}\cdot (150^{2}-150) =  0.5881432 (22350) = 13145
\end{align*}

\begin{table}[!ht]
\centering
\caption{Fisher Iris dataset, \(n = 150\), which allows for \(\alpha_{1},\alpha_{2}\) to be solved for as a function of the known distances and the sample size.}
\label{tab:example_moments}
\tiny
\begin{tabular}{rrrrr}
  \toprule
 & Sepal.Length & Sepal.Width & Petal.Length & Petal.Width \\ 
  \midrule
Sepal.Length & 0 & \(\alpha_{1} = 0.5285905, \alpha_{2} = 0.4567304\) & \(\alpha_{1} = 0.151408205, \alpha_{2} = 0.8409595\) & \(\alpha_{1} = 0.1881198, \alpha_{2} = 0.8028011\) \\ 
  Sepal.Width & 11990 & 0 & \(\alpha_{1} = 0.5797333, \alpha_{2} = 0.4059677\) & \(\alpha_{1} = 0.5646111, \alpha_{2} = 0.4209448\) \\ 
  Petal.Length & 3410 & 13145 & 0 & \(\alpha_{1} = 0.1171291, \alpha_{2} = 0.8767339\) \\ 
  Petal.Width & 4243 & 12804 & 2634 & 0 \\ 
   \bottomrule
\end{tabular}
\end{table}

\subsection{Maximum likelihood estimator}
\label{subsec:mle}
The method of moments estimator does not make explicit use of the likelihood function defined by the probability, which loses certain benefits. However, it has been historically observed to perform preferentially, relative to the MLE, upon Beta distributed independent random variables. While the construction of order statistics has been historically defined upon continuous random variables, with constant observed variance, we introduce, using the Kemeny metric, a framework for which the shape  parameters for each marginal random variable distributions \(x\) and \(y\) may be obtained by maximum likelihood estimators. From these parameters, we also examine the MLE for the Kemeny correlation statistic resulting from the product of the two independent Beta random variables to produce a maximum likelihood estimator. Of particular note is that, due to the removal of the assumption of the absence of ties (due to applicability to both discrete and continuous random variables), the variance is no longer an explicit a priori known, and thus is explicitly estimated upon the data by the obtained shape parameters, providing a MLE estimator of the variance as well. 
% The product of the two random Beta distributions is taken by application of the Mellin transformation to each random variable, using the well known property that the product of the Mellin transformations of two random variables is equal to the Mellin transformation of the product itself. 
This allows us to shown that the likelihood function of the product of two random variables is equivalent to the Beta distribution of the Kemeny metric, and thus is itself a maximum likelihood estimator. We provide a framework for the construction of the Beta-distributed (normalised) extended real random variables \(X_{j}\) of length \(n\), and demonstrate that the mean and variance of said variables are directly obtained by known constructions of the moments for the Beta distributions from the estimated shape parameters.  

% \begin{equation}
% \begin{aligned}
% \ln \,{\mathcal {L}}(\alpha_{1} ,\alpha_{2} \mid X) & = \sum _{i=1}^{N}\ln \left({\mathcal {L}}_{i}(\alpha_{1} ,\alpha_{2} \mid X_{i})\right)\\
%       & = \sum _{i=1}^{N}\ln \left(f(X_{i};\alpha_{1} ,\alpha_{2} )\right)\\
%       & = \sum _{i=1}^{N}\ln \left({\frac {X_{i}^{\alpha_{1} -1}(1-X_{i})^{\alpha_{2} -1}}{\mathrm {B} (\alpha_{1} ,\alpha_{2} )}}\right)\\
%       & = (\alpha_{1} -1)\sum _{i=1}^{N}\ln(X_{i})+(\alpha_{2} -1)\sum _{i=1}^{N}\ln(1-X_{i})-N\ln \mathrm {B} (\alpha_{1} ,\alpha_{2} )
% \end{aligned}
% \end{equation}

The transformation of the random variables \(X_{i,j}\) onto the order-statistic space \((X_{(i,j)})\) for \(i = 1,\dots, n\) and \(j = 1,\ldots,p\), which is easily confirmed to be a Beta-binomial distribution in the support \([-\tfrac{N}{2},\dots,0,\dots,\tfrac{N}{2}]\) to be Beta-distributed for each independently distributed random variable, are obtained the the extremum of the scores from the vectorisation upon the \(\kappa(X_{j})\). By summing each skew-symmetric matrix \(\kappa\) matrix over the \(n\) rows we produce an \(1 \times n\) vector, presenting an ordered set of values \(\vec{X}_{(j)}\) for random variables \(j\) with the previously defined support. The range in the absence of ties (as arises for continuous random variables) is \([(1-n),n-1] = 2(n-1)\), however the presence of the freely estimated (data dependent) variance allows for such ties to occur, and therefore expands the necessary identification to two parameters for each order statistic distribution of \(X_{j}\). As Beta distributed random variables are assumed to be non-negative \(x \in [0,1)\), we therefore are also required to subtract the minimum order statistic upon \(\vec{X}_{(j)}\) upon each vector (by the symmetric even \(\kappa\) function, this also allows the largest observed sample value to identify the Beta distribution). Before transformation however \(E(\vec{X}_{(j)}) = 0,\) and is therefore observed to be unbiased, just as the underlying Kemeny metric itself is unbiased (Lemma~\ref{lem:unbiased}). Further inspection confirms that the analytic solution to the variance of \(\vec{X}\) and the \(\sigma^{2}_{\kappa}(X)\) are equivalent, thereby confirming the duality of the function representation: 

% \begin{subequations}
% \footnotesize
% \begin{multicols}{2}
% \begin{equation}
% \vec{X}_{(j)}  = \Big(\sum_{k=1}^{n} \kappa_{kl}(X_{j})\Big)^{\intercal}, ~1 \le l \le n \\
% \end{equation}
% \begin{equation}
% \vec{\xi}_{j}  = \| |X_{(j)}|\|_{\infty}\\
% \end{equation}
% \begin{equation}
% \label{eq:z_sigma}
% \vec{\sigma}_{j}^{2}  = \frac{1}{2\cdot\vec{\xi}_{j}\cdot (n-1)}\sum_{i=1}^{n} (X_{i,(j)})^{2}\\
% \end{equation}
% \begin{equation}
% \label{eq:z_spearman}
% \vec{Z}_{(j)}  = \tfrac{1}{2\cdot\vec{\xi}_{j}} (\vec{X}_{(j)} + \vec{\xi}_{j}),~ 0 \le \vec{X}_{(j)} \le 1.
% % \end{aligned}
% \end{equation}
% \end{multicols}
% \end{subequations}
This produces a set of \(j\) Beta distributed random variables which are independently distributed and non-degenerate. From these, a set of \(j\) pairs of parameters, \((\alpha_{(1,j)},\alpha_{(2,j)})\), may be estimated by method of moments or maximum likelihood, providing unique estimates for the first and second moments (uniqueness defined by the compact and totally bounded nature of the Beta distributed variables). This also provides a natural construction for the standardised sample characterisations of the variables \(Z_{(j)} = \tfrac{X_{(j)}}{\sqrt{\vec{\sigma}^{2}_{j}}}\), which we note for the reader is the groundwork for the construction of Spearman's \(\rho\), generalised such that it is identified (i.e., measurable) in the presence of ties. It is also a signed measure, unlike existing standard construction for the Wald statistics as applied to non-parametric data analysis. 

The inner product of these two Beta random variables \(\binom{p}{2} = \{k,k^{\prime}\} = K\), resulting from equation~\ref{eq:beta_nonneg} (i.e., the lower diagonal elements of the variance-covariance and correlation matrices). The Kemeny spanning basis and its isometric Euclidean mapping naturally follows as the correlation coefficient and the partial Wald test statistic, estimated upon the sample, with unit variance upon the population. Note however that \(\vec{\sigma}_{j}^{2}\) is both the empirical representation of the almost surely finite variance for all \(n\), and also that by equations \ref{eq:moment1} \& ~\ref{eq:moment2}. Therefore, we possess the just-identified maximum likelihood estimators, who inner-product is then Spearman's \(\rho\) correlation, and therefore also an unbiased maximum likelihood estimator, centred at expectation of 0, with a probability distribution defined as the product of two Beta distributions \(\langle{x,y\rangle}\). 

\begin{theorem}
The Kemeny \(\rho\) correlation coefficient (which in the limiting case is equivalent to the Spearman \(\rho\) upon continuous independent random variables) is a maximum likelihood estimator upon the Kemeny metric topological function space.
\end{theorem}
\begin{proof}
Begin with the marginal distribution evaluation for each random variable \(x_{k},y = x_{k^{\prime}}\), denoting two \(\kappa\) mappings from independently and identically sampled finite distributions upon the extended reals, estimated by standard MLE, for which the Gaussian errors are asymptotically correct. Allow \(\vec{Z} = \langle {x}{y}\rangle\) as a function of the marginal parameters \(\alpha_{(1,x)},\alpha_{(1,y)}\), \(\alpha_{(2,x)},\alpha_{(2,y)}\). These result from the random variables being transformed by the \(\kappa\) mapping under equation~\ref{eq:beta_nonneg}, producing strictly non-negative values in the unit interval. The first and second moments of \(Z,\mu^{1}_{\vec{Z}},\mu^{2}_{\vec{Z}}\) are constructed using the 2 parameters for each random variable under the following estimating equations:
\begin{subequations}
\footnotesize
\label{eq:mle_spearman}
\begin{multicols}{2}
\begin{equation}
\mu_{\vec{x}}^{1}  = \frac{\alpha_{1,\vec{x}}}{\alpha_{1,\vec{x}} + \alpha_{2,\vec{x}}}
\end{equation}
\begin{equation}
\mu_{\vec{y}}^{1}  = \frac{\alpha_{1,\vec{y}}}{\alpha_{1,\vec{y}} + \alpha_{2,\vec{y}}}
\end{equation}
\begin{equation}
\mu_{\vec{x}}^{2}  = \frac{\alpha_{1,\vec{x}} \alpha_{2,\vec{x}}}{(\alpha_{1,\vec{x}}+\alpha_{2,\vec{x}})^{2} \cdot (\alpha_{1,\vec{x}} + \alpha_{2,\vec{x}} + 1)}
\end{equation}
\begin{equation}
\mu_{\vec{y}}^{2}  = \frac{\alpha_{1,\vec{y}} \alpha_{2,\vec{y}}}{(\alpha_{1,\vec{y}}+\alpha_{2,\vec{y}})^{2} \cdot (\alpha_{1,\vec{y}} + \alpha_{2,\vec{y}} + 1)}
\end{equation}
\end{multicols}
\normalsize
\begin{equation}
\mu_{\vec{Z}}^{1}  = \Big(\tfrac{\hat{\alpha}_{1,x}}{\hat{\alpha}_{1,x} + \hat{\alpha}_{2,x}}\Big)\cdot \Big(\tfrac{\hat{\alpha}_{1,y}}{\hat{\alpha}_{1,y} + \hat{\alpha}_{2,\vec{y}}}\Big) = \mu_{\vec{x}}^{1} \cdot \mu_{\vec{y}}^{1}
\end{equation}
\begin{equation}
\label{eq:mle_prob}
\mu_{\vec{Z}}^{2}  = \Big(\tfrac{\alpha_{1,\vec{x}} \alpha_{2,\vec{x}}}{(\alpha_{1,x}+\alpha_{2,\vec{x}})^{2} \cdot (\alpha_{1,\vec{x}} + \alpha_{2,\vec{x}} + 1)}\Big) \cdot \Big(\tfrac{\alpha_{1,\vec{y}} \alpha_{2,\vec{y}}}{(\alpha_{1,\vec{y}}+\alpha_{2,\vec{y}})^{2} \cdot (\alpha_{1,\vec{y}} + \alpha_{2,\vec{y}} + 1)}\Big)  = \mu_{\vec{x}}^{2} \cdot \mu_{\vec{y}}^{2} + (\mu_{\vec{x}}^{1})^{2} \cdot (\mu_{\vec{y}}^{1})^{2}
\end{equation}
\end{subequations}

Constructively then, take two extended real vectors \(X,Y\) which are independently and identically sampled. Then the following equivalence is obtained 
\begin{subequations}
\begin{multicols}{2}
\begin{equation}
\label{eq:beta_nonneg}
\vec{x} = \sum_{k=1}^{n} \kappa_{kl}(x) - \min{\sum_{k=1}^{n} \kappa_{kl}(x)}\\%}{\tfrac{1}{\sqrt{2(n-1)}}\sum_{k=1}^{n}\kappa^{2}_{kl}(X)}\\
\end{equation}
\begin{equation}
\vec{y} = \sum_{k=1}^{n} \kappa_{kl}(y) - \min{\sum_{k=1}^{n} \kappa_{kl}(y)}\\%}{\tfrac{1}{\sqrt{2(n-1)}}\sum_{k=1}^{n}\kappa^{2}_{kl}(X)}\\
\end{equation}
\begin{equation*}
\sigma^{2}_{(x)} = \tfrac{1}{2(n-1)}\sum_{k=1}^{n}\sum_{l=1}^{n} \kappa^{2}_{kl}(x)\\
\end{equation*}
\begin{equation*}
\sigma^{2}_{(y)} = \tfrac{1}{2(n-1)}\sum_{k=1}^{n}\sum_{l=1}^{n} \kappa^{2}_{kl}(y)\\
\end{equation*}
\end{multicols}
\begin{equation}
\label{eq:analytic}
\rho_{(x,y)} = \tfrac{1}{n-1} \sum_{i=1}^{n} \frac{\vec{x}_{i}\cdot\vec{y}_{i}}{\sqrt{\sigma^{2}_{(x)}}\sqrt{\sigma^{2}_{(y)}}} = \tfrac{n}{n-1}\frac{\mu_{\vec{Z}}^{1} - \mu_{\vec{x}}^{1}\cdot\mu^{1}_{\vec{y}}}{\sqrt{\mu^{2}_{\vec{y}}\cdot\mu^{2}_{\vec{y}}}}\\
\end{equation}
\end{subequations}
% ZZ <- data.frame(MASS::mvrnorm(n = 1000,mu = c(2,-50),Sigma = matrix(c(1,.75,.75,1),ncol = 2)))
%   colnames(ZZ) <- c('X','Y')
%   X <- (colSums(kemenyscore(ZZ$X)) - min(colSums(kemenyscore(ZZ$X))))*sqrt(.5)
%   Y <- (colSums(kemenyscore(ZZ$Y)) - min(colSums(kemenyscore(ZZ$Y))))*sqrt(.5)
%      (mean(X * Y) - mean(X)*mean(Y))/
%         (sd(X) * sd(Y)) * (n/(n-1))
Note that the empirical MLE estimated marginal variable parameters are produced to calculate the mean and variance of the inner-product, with the necessary Bessel correction, which under the limit \(\tfrac{n}{n-1} \to 1\) and therefore equation~\ref{eq:analytic} and equation~\ref{eq:mle_prob} are equivalent upon the population. It is thereby proven that the probability distribution of \(Z = xy\) when evaluated as a normed Beta distribution of the original order statistics produced independent of the presence or absence of ties, produces Kemeny's \(\rho\). Said estimator is the maximum likelihood and minimum variance estimator for said co-product, as a linear function of the MLE parameters, thereby concluding the proof. Unbiasedness directly follows as a linear operation upon the \(\kappa\) functions, which are themselves unbiased estimators.
\end{proof}

\begin{remark}
The variances of random variables measured upon the Kemeny \(\rho\) and Kemeny \(\tau_{\kappa}\) function  are identical, thereby allowing for identification upon all finite \(n\).  Further, the maximum likelihood estimation of the Kemeny \(\tau_{\kappa}\) correlation coefficient may then, by the monotonic orthonormal transformation of the MLE Kemeny \(\rho\) parameter obtained using equation~\ref{eq:kendall_sin}, also be considered a maximum likelihood estimator.
% \end{corollary}
\begin{proof}[Remark 3]
Any function of maximum likelihood estimators is itself a maximum likelihood estimator, by functional equivalence. The crossproduct and variances are themselves maximum likelihood estimators (as given in equations~\ref{eq:mle_spearman}), and the Kemeny \(\rho\), as an affine linear transformation of said parameters, is itself an unbiased maximum likelihood estimator. Note however that this maximum likelihood estimator is constructed upon the domain of \(n\) bivariate \(i.i.d.\) samples, rather than \(\mathcal{M}^{*}\). By the bijection of the formulation, the maximum likelihood estimator, transformed via equation~\ref{eq:kendall_sin} for which the derivative \(\tau_{\kappa}\) is taken, is conjointly maximised by definition at the same stable point, as required by the projective duality, thereby ensuring that the maximum likelihood solution is obtained at both metric domains with a common solution both in the limit and for finite samples. Tighter bounds are further observed, due to the existence of a two orthonormal linear equation criterion, at least one of which is almost surely a Gauss-Markov estimator upon all finite samples. 
\end{proof}
\end{remark}

% \subsection{Duality between the Kemeny \(\tau_{\kappa}\) and Spearman's \(\rho\)}

% \begin{remark}[Studentising the Kendall \(\rho\) estimator]
% If we accept the scenario of a common variance population between two independent random variables of length \(n\), then a partial Wald statistic exactly equivalent to \(z_{\phi}\) must result, as all parameters are equivalent and dual. However, consider the more common scenarios, wherein the variance is independently estimated upon either one or two distinct samples represented by the original sampled vectors \(X,Y\). Then, as we no longer require the variance of each vector to arise from \(S_{n}\), the variances may be freely estimated, and as such results in a reduction on the degrees of freedom, corresponding to \(t_{n-1}\) and \(t_{n-2}\), as appropriate. 
% \end{remark}

\subsection{Point bi-serial correlations}
Standard non-parametric estimators are also applied to group-wise bivariate correlations drawn from a common population, most familiarly termed the Wilcoxon-Mann-Whitney test. However, estimation in the presence of ties upon the outcome variable has been complicated by the lack of established \(\sigma\)-algebra, which has now been defined. This corresponds to the existence of distinct estimation procedures for the Wald and general probability frameworks under the Null Hypothesis Significance Testing approaches to non-parametric testing, as approached for correlations. By the strict sub-space characterisation of the common domain between the estimators as proposed in this manuscript, it immediately follows that performance is identical in the absence of ties. However, when present, the unbiased and minimum variance properties of our proposed estimators apply equivalently to the problem of location difference shifts as well. For example, the point-biserial distance from the null hypothesis that the Kemeny distance is equal to the median (i.e., independence between the group and outcome) occurs with probability \(\Pr_{z_{\phi}}(-\tfrac{-49}{30.63658}) = 0.1097329\), in line with the findings of the Pearson correlation, the Spearman correlation, the Kendall's \(\tau_{b}\) correlation, and the Wilcoxon Rank-Sum test with continuity correction (W = 25.5, p = .06933). Using the conventional \(\alpha\)-decision threshold, the null hypothesis still fails to be rejected, however as shown in our results for the next section, the presence of ties results in a less powerful (more widely dispersed) estimator upon the non-parametric space. 

\subsubsection{Numerical demonstration}
Under bootstrap re-sampling, we demonstrate that our proposed estimator is the most powerful estimator and exhibits the desired empirical minimum variance characterisation amongst all competitors, as constructed for a standard non-parametric two-sample location difference test. The data set was taken from the \(20 \times 3\) canonical `sleep' data set available in R, performed using repeated sampling to produce 25,000 150 sized data sets, with the test statistic results presented in Table~\ref{tab:non_par-location} and various non-parametric effect sizes derived from the correlation coefficients presented in Table~\ref{tab:non_par-effect}. As demonstrated with these results, we observe that the (known to be biased) non-parametric estimators have larger than accurate first order approximations (means) and second order approximations (standard deviations), exactly in line with the performance expected for any James-Stein (i.e., a biased or non-minimum variance) style estimator. The ranges of the competing estimators and the median absolute deviations are all, as hypothesised, found to be larger as well. These are provided to confirm that the minimum variance property of the partial Wald tests also (by definition) holds for the estimators, and also that standard correlation effect sizes have a direct relationship to the Kemeny metric, a given necessity by the projective geometric duality of almost any estimation (Statistical Learning) problem.

Investigation of the numerical results provide direct confirmation of that all measures of dispersion are minimised by the Kemeny based estimator class (e.g., range, mad and sd). Further, we observe that the non-normality of the class of estimators' test statistics is minimised most effectively for the Kemeny parameters, thereby demonstrating the most normal unbiased and minimum variance performance characteristics which satisfies the expected Gauss-Markov normality of the generalised central limit theorem, without concern to the objective linearity decomposition of the scores themselves (i.e., the \textcite{nelder1972} monotonic function space). In particular, we note that this allows for the analysis of a much more convenient ordinal linear framework for rating scales

%           mean    sd    median trimmed mad min max range skew kurtosis
% kemeny  4.44614 0.72600 4.47410 4.45843 0.72044 .903600 6.92161 6.01802 -0.18223 -0.01279
% wilcox  1425.74502 223.31724  1417   1421.95175   222.39 664.5   2516.50000  1852  0.18139 -0.01912
% kendall  5.17413 0.84323 5.20773 5.18852 0.83661 1.05255 8.04154 6.98899 -0.18372 -0.01263
% spearman 324052.53952 38855.63132 322504.23434 323389.13939 38550.80085 1.91923e+05 513973.66622 322050.70382  0.18372 -0.01263
% pearson 5.36339     1.03554      5.35087      5.35422     1.01913 1.23854      9.91650      8.67795  0.09475  0.02065

\begin{table}[!ht]
\tiny
\centering
\caption{Comparison of the two-sample difference of location estimators' test statistics, to evaluate power and performance upon the Sleep data set for various sample sizes}
\label{tab:non_par-location}
\hspace*{-2.0cm}
\begin{tabular}{ccccccccccc}
  \toprule
Sample size & &  mean & sd & median & mad & min & max & range & skew & kurtosis  \\ 
  \midrule
\multirow{5}{*}{20} & Kemeny & 1.51678 & 0.71893 & 1.56675 & 0.72590 & -1.99108 & 3.26407 & 5.25516 & -0.43512 & 0.12202 \\ 
  & Wilcox & 24.26342 & 11.00131 & 23.5 & 11.1195 & 0 & 80 & 80 & 0.47789 & 0.16407 \\ 
  & Kendall & 1.81228 & 0.85218 & 1.87730 & 0.84998 & -2.34019 & 3.85763 & 6.19782 & -0.47474 & 0.16679 \\ 
  % &   Spearman & 777.03245 & 260.01918 & 757.19170 & 259.34822 & 152.94786 & 2044.04418 & 1891.09632 & 0.47474 & 0.16679 \\ 
  &   Pearson & 1.94842 & 1.09417 & 1.89843 & 1.04310 & -2.58812 & 8.59658 & 11.18469 & 0.36155 & 0.74704 \\ 
\midrule
\multirow{5}{*}{150} &  Kemeny & 4.446 & 0.726 & 4.47 & 0.72 & 0.90 & 6.92 & 6.01802 & -0.18223 & -0.01279 \\ 
 &  Wilcox &     1425.745 & 223.317 & 1417 & 222.39 & 664.50 & 2516.50 & 1852 & 0.18139 & -0.01912 \\ 
 &  Kendall &     5.174 & 0.843 & 5.21 & 0.84 & 1.05 & 8.04 & 6.98899 & -0.18372 & -0.01263  \\ 
 % &  spearman &    324052.54 & 38855.631 & 322504.23 & 323389.14 & 191922.96 & 513973.67 & 322050.70 & 0.18 & -0.01\\ 
 &  Pearson & 5.363 & 1.03554 & 5.35 & 1.02 & 1.24 & 9.92 & 8.67795 & 0.09475 & 0.02065 \\ 
\midrule
\multirow{5}{*}{750} & Kemeny & 10.03460 & 0.73047 & 10.04483 & 0.73596 & 6.42783 & 12.68464 & 6.25682 & -0.08382 & 0.02329 \\ 
  & Wilcox & 35833.29776 & 2503.45323 & 35803.00000 & 2521.90260 & 26834.00000 & 48158.50000 & 21324.50000 & 0.08592 & 0.02627 \\ 
  & Kendall & 11.62749 & 0.84565 & 11.63919 & 0.85096 & 7.45119 & 14.68909 & 7.23790 & -0.08438 & 0.02387 \\ 
  % & spearman & 40439506.09023 & 2172596.76456 & 40409462.19572 & 2186243.00621 & 32573776.13607 & 51169090.65989 & 18595314.52382 & 0.08438 & 0.02387 \\ 
  & Pearson & 12.00136 & 1.02885 & 11.99176 & 1.03368 & 7.33255 & 16.16491 & 8.83236 & 0.06177 & 0.02354 \\ 
\bottomrule
\end{tabular}
\end{table}
 
\begin{table}[!ht]
\scriptsize
\centering
\caption{Comparison of the two-sample difference of location estimators' effect sizes, to evaluate power and performance upon the Sleep data set for various sample sizes}
\label{tab:non_par-effect}
\hspace*{-1.0cm}
\begin{tabular}{ccccccccccc}
  \toprule
Sample size & &  mean & sd & median & mad & min & max & range & skew & kurtosis  \\ 
  \midrule
\multirow{7}{*}{20} & Kemeny & 0.25789 & 0 & 0.25789 & 0 & 0.25789 & 0.25789 & 0 & ---  & --- \\ 
  & \(\sin(\text{Kem}~\tau_{\kappa})\) & 0.39411 & 0 & 0.39411 & 0 & 0.39411 & 0.39411 & 0 & ---  & --- \\ 
  &   Wilcox r & -0.41689 & 0.19097 & -0.44700 & 0.17717 & -0.78900 & 0.26200 & 1.05100 & 0.72752 & 0.22675 \\ 
  &   Glass' r & -0.50094 & 0.23034 & -0.53100 & 0.20534 & -0.94800 & 0.33300 & 1.28100 & 0.72624 & 0.26976 \\ 
  &   Kendall & 0.36761 & 0.16849 & 0.39367 & 0.15338 & -0.23295 & 0.69614 & 0.92908 & -0.72588 & 0.23961 \\ 
  &   Spearman & 0.42775 & 0.19591 & 0.45893 & 0.18081 & -0.26784 & 0.81054 & 1.07838 & -0.72792 & 0.22655 \\ 
  &   Pearson & 0.40385 & 0.18343 & 0.42146 & 0.17135 & -0.22409 & 0.74867 & 0.97276 & -0.73520 & 0.61920 \\ 
\midrule
\multirow{7}{*}{150} & Kemeny & 0.24500 & 0.03954 & 0.24635 & 0.03940 & 0.08116 & 0.38917 & 0.30801 & -0.19257 & 0.03618 \\ 
  & \(\sin(\text{Kem}~\tau_{\kappa})\) & 0.37469 & 0.05760 & 0.37739 & 0.05738 & 0.12715 & 0.57394 & 0.44680 & -0.26570 & 0.08003 \\ 
  & Wilcox r & -0.42277 & 0.06812 & -0.42500 & 0.06820 & -0.67000 & -0.14000 & 0.53000 & 0.19522 & 0.03770 \\ 
  & Glass' r & -0.48998 & 0.07898 & -0.49300 & 0.07858 & -0.77300 & -0.16100 & 0.61200 & 0.19604 & 0.03656 \\ 
  & Kendall & 0.35828 & 0.05761 & 0.36033 & 0.05751 & 0.11875 & 0.56640 & 0.44765 & -0.19380 & 0.04031 \\ 
  & Spearman & 0.42419 & 0.06834 & 0.42672 & 0.06824 & 0.14030 & 0.67137 & 0.53106 & -0.19526 & 0.03775 \\ 
  & Pearson & 0.40066 & 0.06396 & 0.40293 & 0.06354 & 0.10657 & 0.64143 & 0.53486 & -0.17317 & 0.03530 \\  
\midrule
\multirow{7}{*}{750} & Kemeny & 0.24501 & 0.01779 & 0.24532 & 0.01776 & 0.17397 & 0.31319 & 0.13922 & -0.08538 & 0.00562 \\ 
  & \(\sin(\text{Kem}~\tau_{\kappa})\) & 0.37529 & 0.02591 & 0.37588 & 0.02584 & 0.26988 & 0.47235 & 0.20247 & -0.11913 & 0.01306 \\ 
  & Wilcox r & -0.42487 & 0.03085 & -0.42700 & 0.03262 & -0.54400 & -0.30200 & 0.24200 & 0.08827 & 0.01076 \\ 
  & Glass' r & -0.49003 & 0.03558 & -0.49100 & 0.03558 & -0.62600 & -0.34800 & 0.27800 & 0.08577 & 0.00614 \\ 
  & Kendall & 0.35834 & 0.02593 & 0.35877 & 0.02591 & 0.25479 & 0.45754 & 0.20275 & -0.08526 & 0.00694 \\ 
  & Spearman & 0.42515 & 0.03085 & 0.42569 & 0.03084 & 0.30184 & 0.54302 & 0.24119 & -0.08581 & 0.00609 \\ 
  & Pearson & 0.40153 & 0.02876 & 0.40192 & 0.02873 & 0.28588 & 0.50723 & 0.22135 & -0.08318 & 0.03151 \\ 
\bottomrule
\end{tabular}
\end{table}

There are a number of noteworthy characteristics to highlight in these results. In particular, stemming from Table~\ref{tab:non_par-effect}, which provides comparable measures (i.e., effect sizes), demonstrates several important flaws with the existing Kendall \(\tau_{b}\) estimator, in that the average biased correlation coefficient is closer to the Pearson correlation than to the Kemeny correlation. However, the Kemeny correlation \((\tau_{\kappa})\) and its sinusoidal transformation (analogous to Spearman's \(\rho\)) display smaller variances (as well as ranges and MAD; Table~\ref{tab:non_par-effect}) than all other estimators, even considering these statistics were all designed to be non-parametric estimators, which are visually depicted in Figure~\ref{fig:density_correlaton_wilcox}.

\begin{figure}[!ht]
% \centering
\hspace*{-2.0cm}
\caption{Visualisation of the distributions of various point-biserial correlation coefficients upon 25,000 resampled sets of \(n=750\) elements.}
\label{fig:density_correlaton_wilcox}
\includegraphics[height = 8cm,width = 15cm]{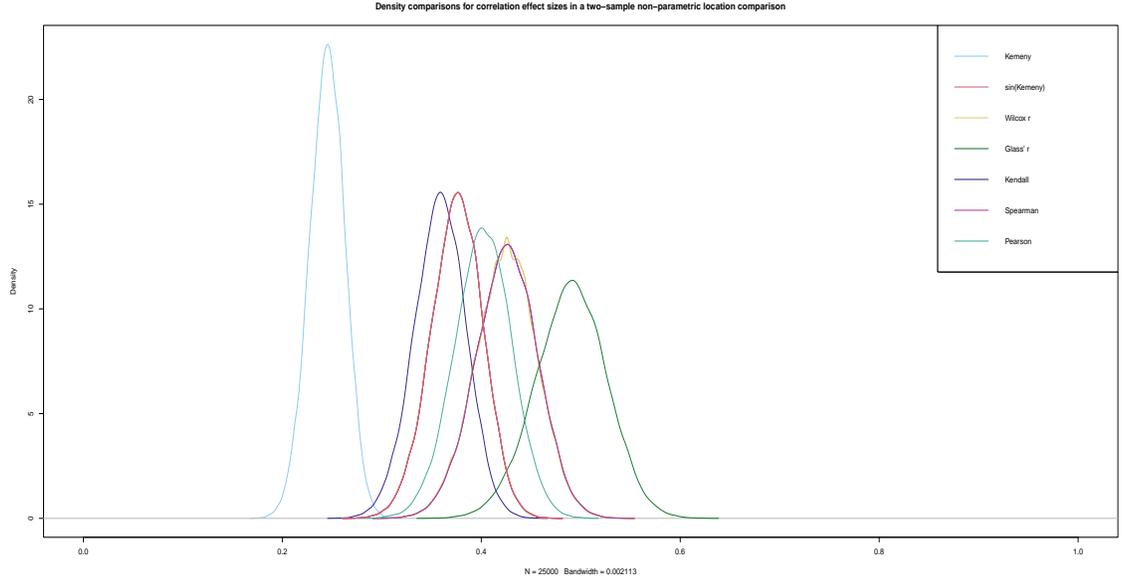}
\end{figure}

While it is true that largest effect size estimators are viable, and even desirable, candidates, it is clearly empirically and theoretically derived that even under a completely non-parametric bootstrap re-sampling procedure, the tightest bounds are obtained using the parametric framework for arbitrary distributed sample proposed in this manuscript. We also would highlight that the sinusoidal transformation of the Kemeny \(\tau_{\kappa}\) correlation possesses a central estimator which is, qualitatively, the average of the alternative estimators available, including the model misspecification found in the use of the Pearson correlation. This provides an excellent demonstration of the duality principle while further demonstrating that the presence of ties substantially affects the estimation characteristics of non-parametric rank-based estimators. The regularity of the Kemeny topology is functionally more explicit, and thus provides the highly desired tightest bounds which are otherwise only empirically estimable, and will only asymptotically achieve the finite sample Cramer-Rao lower-bound which our estimator already is proven to possesses, and which is achieved using the most `normal' approximation of the skewness and excess kurtosis expected values (both asymptotically 0 for any linear function space; Lemma~\ref{lem:cramer-rao}). 
 
\subsection{The introduction of greater sample dependency (Studentisation) in the Kemeny linear model}
\label{subsection:studentising}
A natural question now arises, when encountering parametric distributions of a statistical estimator which satisfies the Gauss-Markov theorem upon any bivariate set of independent random variables: how does one construct the test statistic to examine for significant differences from 0, the null hypothesis, upon a finite sample, when \(\mathcal{M}\) may not be sampled without replacement? The previously derived representation provides us the means to avoid the bifurcation of the variance approximation between tied and non-tied samples, and allows for an exact p-value to be uniquely determined for all finite samples, using the distribution of distances (affine linear transformations of the respective distances for the correlations) divided the standard deviations of these distributions.

This representation however is not exhaustive, as it precludes the existence of, for instance, finite response sets, in which the rank ordering (and thus the standard deviation) is less than the extrema, and which is also not always uniform. This allows us to construct estimators upon domains which encompass more degrees of freedom: in essence, this allows us to extend the characterisation of the t-distributions to non-parametric estimation scenarios. Pragmatically, this provides a means of estimation for the Kemeny estimator for which the Student's and Welch's t-tests may arise. We provide an empirical approximation by a scaling nuisance parameter upon the obtained \(t\)-statistics, to maintain the compact and totally bounded nature of the Kemeny Wald statistics. Considering however that the variances of Kemeny's \(\rho\) and equation~\ref{eq:kem_variance} have already been noted to be equivalent (as variances of the same linear \(\kappa\) function mapping), this allows for both a limiting one-sample testing procedure with known population variance (and arbitrary realised distance) and also for finite sample corrections to the estimated common sample population.  
% However, we also explore work by \textcite{ord1967}, for which the Pearson \rom{7} distribution is discretised, consistent with the discretised nature of the uniform marginal distribution and Beta-binomial distribution.

\subsubsection{A non-parametric alternative to Student's t-test for one-sample}

One would quite naturally wonder where the degrees of freedom arise from to construct the t-distribution, as compared to the sufficient statistic of the Kemeny distance defined in equation!~\ref{eq:z_kemeny}. When addressing the previously described correlation estimator (see equation~\ref{eq:kem_cor}), the cross-product is already normed, possessing only 1 degree of freedom corresponding to the  distance between two vectors relative to the uniquely defined origin at 0: \(\frac{n^{2}-n}{2} \pm n^{2}-n\), resulting in  \(df = n\). In particular however should be noted the occurrences for which these events arise, namely non-uniform sampling over \(\mathcal{M}\) which naturally must increase the expected variance of the distribution of all Kemeny distances \(\rho_{\kappa}\). Therefore, we desire a sample dependent means of increasing the variances to reflect the greater than otherwise extant probability of sampling certain (odd) distances with greater proclivity, or probability. This then defines the standard scenario for the application of Slutsky's theorem. We now show that this expansion is one sided (i.e., the variance of the \(\rho_{\kappa}\) function space is lower bounded by \(\dot{\sigma}^{2}_{\kappa}(n)\)):
\begin{lemma}
\label{thm:slutsky_lemma}
The sample dependent variance of the population of all Kemeny distances \(\mathcal{M}\) is almost surely greater than or equal to \(\dot{\sigma}^{2}_{\kappa}(n)\) (equation~\ref{eq:kem_population_variance}).
\end{lemma}
\begin{proof}
The support of the signed distance function space for \(\rho_{\kappa}\) is given by the compact and totally bounded interval \(U(\mathcal{M}) \in \mathbb{N}\), with expectation 0, and almost surely finite variance (defined using equation~\ref{eq:kem_population_variance}). As the variance of all distances are measurable upon the population over \(\mathcal{M}\), any reduction in the variance of the unbiased estimator must correspond to a truncation (reduction) of the defined support of the Kemeny distance or a decrease in the set \(\mathcal{M}\), which is exhaustive over all permutations without replacement, or else become a biased estimator. By the Glivenko-Cantelli theorem (Theorem~\ref{thm:gc}), the signed Kemeny distance is uniform, and given its unbiasedness (Lemma~\ref{lem:unbiased}), the only viable alternative is to reduce the support of the measure space in a uniform manner, while remaining an unbiased estimator. 

As the independent random variables measured upon the Kemeny metric of length \(n\) must exhibit a static neighbourhood of distances no less than \(U(\mathcal{M})\), it therefore follows that shrinking the variance must also shrink the support, in order to restrict the possible observable space. Such an operation would naturally result in a non-measurable vectors while remaining uniformly i.i.d. sampled by axiomatic definition, and one would which be, by definition, biased upon the sample. However as the Kemeny metric is unbiased, this results in a paradox. Therefore, we reject the hypothesis that the variance must decrease the support, and it therefore follows, as by proof by contradiction, that the Kemeny variance is lower bounded by equation~\ref{eq:kem_population_variance} as a function of the observed bivariate sample, else risking non-measurable random variables upon the extended real line resulting in a paradox. By induction upon the given sample \(n\), wherein all \(n < n + k, k = 1,\dots,\infty^{+}\), it follows that the sub-sample data elements in \(k\) are those directly observed upon the distances, thereby introducing no additional information without a corresponding larger sample, which must possess a strong regular probability structure for any sample size (Lemma~\ref{lem:haar}) or else be non-measurable for both larger and smaller samples. 
\end{proof}

More complex estimation scenarios naturally arise however, when dealing with location upon the distribution with a fixed distance, thereby introducing more specific information in the form of vector specific identification. That is, variances of the two independent vectors \(x,y \in \overline{\mathbb{R}}^{n \times 2}\) constructed from the respective \(\kappa(x),\kappa(y),\) per equation~\ref{eq:kem_variance}, such as when dealing with non-continuous random variables. This scenario also addresses the problem of finite sample cardinality \(m \in \mathcal{M}\), as we may now identify specific distribution of observations with frequency as a function of the variances (i.e., repeated, or non-mutually exclusive, sampling), rather than the uniform and mutually exclusive selection upon the exponential family. 

% This more clearly allows the user to depict explicit, yet non-unique (and therefore non-deterministic), locations within the distributions. Within each integer band of Kemeny distances, explicitly ignoring the unique extrema, the variances upon the skew-symmetric matrices resulting from \(\kappa\) more clearly identify the appropriate vector of scores from which our data model arises. From an information theoretic point of view then, entropy is decreased towards 0 when both the distance and the respective variances of the vectors which produce said distances are known. More simplistically, consider that by sub-additivity, the distances from the origin for each vector, in addition to the respective variances, more concretely identify which element in the permutation space \(\mathcal{M}\) identify the data, conditional only upon the variances and distance as sufficient statistics. 

As such, estimating the variance (and proportionally the standard deviation) of each explicitly characterises the distribution which produces a given test, with a sample independent measure space upon which the concentration of the probability density function, and thus the likelihood function, must follow. With more degrees of freedom, the Student \(t\)-distributions arise upon the Wald statistics, wherein the variances, which are most typically distinct (following sampling without replacement), naturally follows. By our given construction we possess the standard scenario breakdown for Student's one and Welch's two-sample t-distributions of test statistics, wherein repeated sampling may occur, such as when observing finite response sets rather than continuous random variables.

By equation~\ref{eq:kem_variance} observe that for the prototypical Binomial case (in which perfect ordering with ties is observed across two events) the variances under the Euclidean and Kemeny formulations are identical, due to the isometric linear equality of the two Hilbert spaces. Using an arbitrary permutation with ties as the default null hypothesis for the ordering \(\mu_{0}\), the corresponding distance from the empirical sample permutation is in \(U(\mathcal{M})\), normed by the standard deviation. Thus, the ratio of the signed distance to the concentration tends to 0 under the following relational, and the resulting Wald test converges to infinite power with the ability to discriminate between all other distances upon the population, without assuming a fixed population of values \(U(\mathcal{M})\), but instead allowing for repeated sampling upon said neighbourhood; this is a simple expectation of the linear convergence of the strong law of law numbers. Here reflected in the term \(s_{p}\), the pooled standard variance between two variables:
\begin{subequations}
\begin{equation}
t_{n-1}  = -\frac{\rho_{\kappa}}{s_{p}}
\end{equation}
\begin{equation}
s_{p} = \sqrt{\tfrac{\dot{\sigma}^{2}_{\kappa}}{2\cdot\hat{\sigma}^{2}_{\kappa}(X)}} %+ \hat{\sigma^{2}}_{\kappa}(Y)}},
\end{equation}
\end{subequations}
where the terms in the denominator are compact and totally bounded upon the interval \((0,0.5]\), and thus sum in expectation (i.e., the combination of a random variable and itself) is a total less than or equal to 1 in the case of random ties occurring with no greater than random chance upon a bivariate distribution. This produces either a sample dependent variance identical to the population and equal in construction to the Spearman's \(\rho\) partial Wald test, or in the presence of greater variability introduces the necessary Slutsky adjustment. The increased probability of further distances arising upon the population of Kemeny distances is thereby allowed, while remaining compact and totally bounded. With this adjustment is also maintained the asymptotic normality, consistency, and minimum variance characteristics we otherwise expect, for any linearly orderable bivariate distribution. 

As characterised, all distances for all sample sizes are normed to be on the unit interval, and therefore the ratio of the distance on the unit interval divided by the finite and totally bounded estimated variance, which is also compact and totally bounded, is always between \([0,1]\), and therefore is almost surely consistent (ignoring the degenerate empirical distributions). If we characterise the signed distance from the arbitrary null hypothesis permutation (standard practice would denote a distance of 0 from the expectation of the median distance for \(n\), equivalent to a mean difference of 0), we obtain almost surely finite empirical estimates, from which is allowed arbitrary sampling of discrete and continuous independent random variables which are orderable. Note that as a function of \(\hat{s}_{p}^{2}\) then, the population distribution of \(U(\mathcal{M})\) must be adjusted as well: this correction, to reflect the greater regular sampling upon the ties (as the non-ties as a negligible and thus ignorable component of any distribution of sufficient \(n\)), we must expect that the variance of the distances must increase, as would be required for more than the expected positive squared distances occurring. Therefore, a simple adjustment for the population is provided, \(U^{*}(\mathcal{M}) = \tfrac{U(\mathcal{M})}{s^{2}_{p}}\), inflated by the rate of concentration. Therefore the variance upon this new distribution is obtained as the adjustment of equation~\ref{eq:kem_population_variance} normed by \(s^{2}_{p}\), \[s^{2}_{\kappa} = U^{*}(\mathcal{M}) = \tfrac{\frac{(n - 1)^2 (n + 4) (2 n - 1)}{18 n}}{s^{2}_{p}},\] which without loss of generality complies with the requirements as established in Lemma~\ref{thm:slutsky_lemma}, and is immediately recognisable to follow a \(\chi^{2}\) distribution upon the asymptotic population.

\subsubsection{Welch's t-test}

% We provide our proposed test statistic upon the Beta-Binomial distribution \(z_{\phi}\) upon support \[\tfrac{n^{2}-n}{2}\bigg(\tfrac{(n-1)^{2}(n+9)(2n-1)}{18n}\bigg)^{-\tfrac{1}{2}},\]
Consider two vectors samples \((x,y)\) independently drawn from a corresponding permutation mapping representation in \(\mathcal{M}^{*}\), allowing for repeated observations (sampling with replacement). Note that the variance-covariance matrix \(\Xi\) upon the Kemeny metric possesses three such estimators, from which is obtained a just-identified system of equations for an \(\overline{\mathbb{R}}^{n\times 2}, n < \infty^{+}.\) Specifically, these three proportionally linear estimators are: (1)~\(\kappa(x)\kappa^{\intercal}(y) = \rho_{\kappa}(x,y),\) (2)~\(\sigma_{\kappa}(x) = \sqrt{\kappa(x)\kappa^{\intercal}(x)},\) (3)~\(\sigma_{\kappa}(y) = \sqrt{\kappa(y)\kappa^{\intercal}(y)},\) denoting the correlation and standard deviations of x and y, respectively.

Under empirical simulation, we may observe that the distribution of the produced distribution is compliant with a Student t-distribution with \(n-2\) degrees of freedom, thereby allowing us to produce a signed test statistic which conforms to the standard t-distribution (Slutsky's theorem) for non-parametric correlations, leveraging the non-constant, and therefore sample dependent measure of variable variance. Empirical simulations have confirmed that the Wald statistic for the bivariate correlation does follow a \(t_{n-2}\) distribution and Figure~\ref{fig:example_student}\footnote{Legend flips the Student and z distribution iconography.} provides empirical comparison for a Student \(t_{1248}\) distribution with scaled non-centrality parameter equal to the mean of the observed distribution, ensuring a consistent empirical expectation between the theoretical distribution and the empirical Wald test statistics; empirical values from the depicted results are provided in Table~\ref{tab:3}, and demonstrate that in the presence of ties upon random variables, the Kendall's \(\tau_{b}\) estimator is neither efficient nor unbiased upon finite samples. Note here that again, the Kemeny Wald tests possess the minimum variance and smallest range (i.e., tightest) with the most normal distribution, all consistent with the proven performance of an unbiased minimum variance estimator under the Gauss-Markov theorem. Consistency of the general Kemeny class estimator \(\hat{\theta}\) and all affine-linear and monotone transformations thereupon is established by equation~\ref{eq:consistency} under Tchebyshev's inequality, which is true by the stable and unbiased nature of the probability distribution (Theorem~\ref{thm:strict_subgauss},Lemma~\ref{lem:unbiased}), which is observed for the entire valid permutation space \(m\to\mathcal{M}^{*}\):  

\begin{subequations}
\begin{align}
\label{eq:consistency}
\lim_{n\to\infty^{+}} \lim_{m\to\mathcal{M}^{*}} \Pr(E_{m}\Big(\frac{\hat{\rho}_{\kappa}(x, y)}{s_{\kappa}} - \frac{\rho_{\kappa}}{s_{\kappa}}\Big)^{2} = 0) = 1, \\
\Pr((\hat{\theta} - \theta)^{2})\geq \epsilon)=\Pr((\hat{\theta} - \theta)^2\geq \epsilon^2)\leq \frac{E(\hat{\theta} - \theta)^2}{\epsilon^2},\\
\lim_{m\to\infty^{+} \equiv \mathcal{M}^{*}} E_{m}(\hat{\theta} - \theta)^2 = 0,~\forall\;n < \infty^{+}.
\end{align}
\begin{equation}
t_{n-2} = -\frac{\theta(x,y)}{s_{\kappa}}
\end{equation}
\begin{equation}
s_{p} = \sqrt{\tfrac{\dot{\sigma}^{2}_{\kappa}}{\hat{\sigma}^{2}_{\kappa}(x) + \hat{\sigma}^{2}_{\kappa}(y)}},
\end{equation}
\begin{equation}
\label{eq:kemeny_pooled}
s_{\kappa} =  \sqrt{\tfrac{\frac{(n - 1)^2 (n + 4) (2 n - 1)}{18 n}}{s^{2}_{p}}} = \sqrt{\frac{((n - 1)^2 (n + 4) (2 n - 1))\cdot s^{-2}_{p}}{18 n}}  
\end{equation}
\end{subequations}
noting that the combination of the two variances in the denominator maintains the averaging of the two random variables' variances by the law of total expectations. The averaging is accomplished due to the use of the scaling constant of \(\sqrt{.5}\) in each \(\kappa^{2}(\cdot)\) matrix cell in equation~\ref{eq:kem_variance}. In equation~\ref{eq:kemeny_pooled} are the concentration adjustments \(\sigma^{2}_{\kappa}(x),\sigma^{2}_{\kappa}(y)\) which are linearly combined to produce the rate of under-dispersion of the ranks (less than or equal to 1), which then divide the known population variance such that arise an estimate of the difference in ranking permutations between two random vectors \(x\) and \(y\) for given \(n\), which thereby norms the population \(U^{*}(\mathcal{M})\) and remains almost surely finite.
% By the compact and totally bounded nature of the variances, the possibility of extremely divergent variances (i.e., \(\sigma^{2}_{\kappa}(\cdot) \in (0,0.5]~s.t., \sigma_{\kappa}^{2}(x) + \sigma_{\kappa}^{2}(x) \le 1, ~a.s.\)), for any distribution of Kemeny signed distances arising upon any common population of \(n\) ordered elements, is strictly 0, ignoring the \(n\) limiting degenerate distribution which occur with probability 0 by definition upon \(\mathcal{M}^{*}\). 

\begin{figure}[!ht]
% \raggedleft
% \caption{Example comparisons of the distributions reported in Table~\ref{tab:3}}
% \begin{subfigure}[b]{.55\textwidth}
% \raggedleft
\centering
\caption{\small {Demonstration that the \(t_{1248}\)-distribution for the Kemeny partial Wald test is consistent to the theoretical distribution in accordance with the necessary Gauss-Markov theorem and has higher-order moments more normally distributed.}}
\label{fig:example_student}
\includegraphics[height = 6cm,width = 10cm,keepaspectratio]{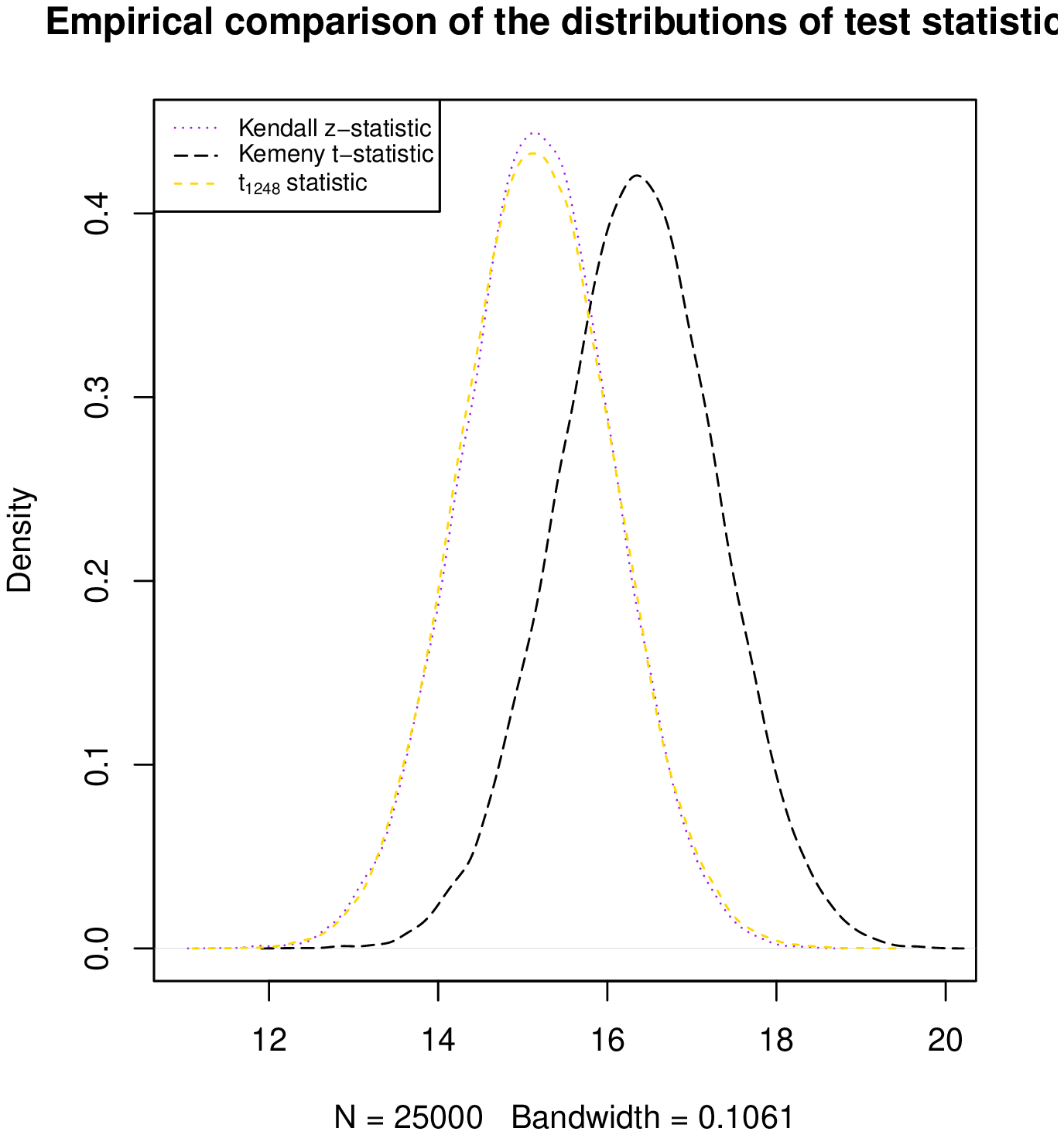}
% \end{subfigure}
% \hfill
% \begin{subfigure}[b]{.45\textwidth}
% \centering
\end{figure}
\begin{figure}[!ht]
\centering
\caption{\small{Visualisation of the empirical distributions with n = 2500}, wherein even for large samples, we observe that the Kemeny \(\tau_{\kappa}\) estimator is both more consistent and more normally distributed.}
\label{fig:example_student2}
\includegraphics[height = 8cm,width = 8cm,keepaspectratio]{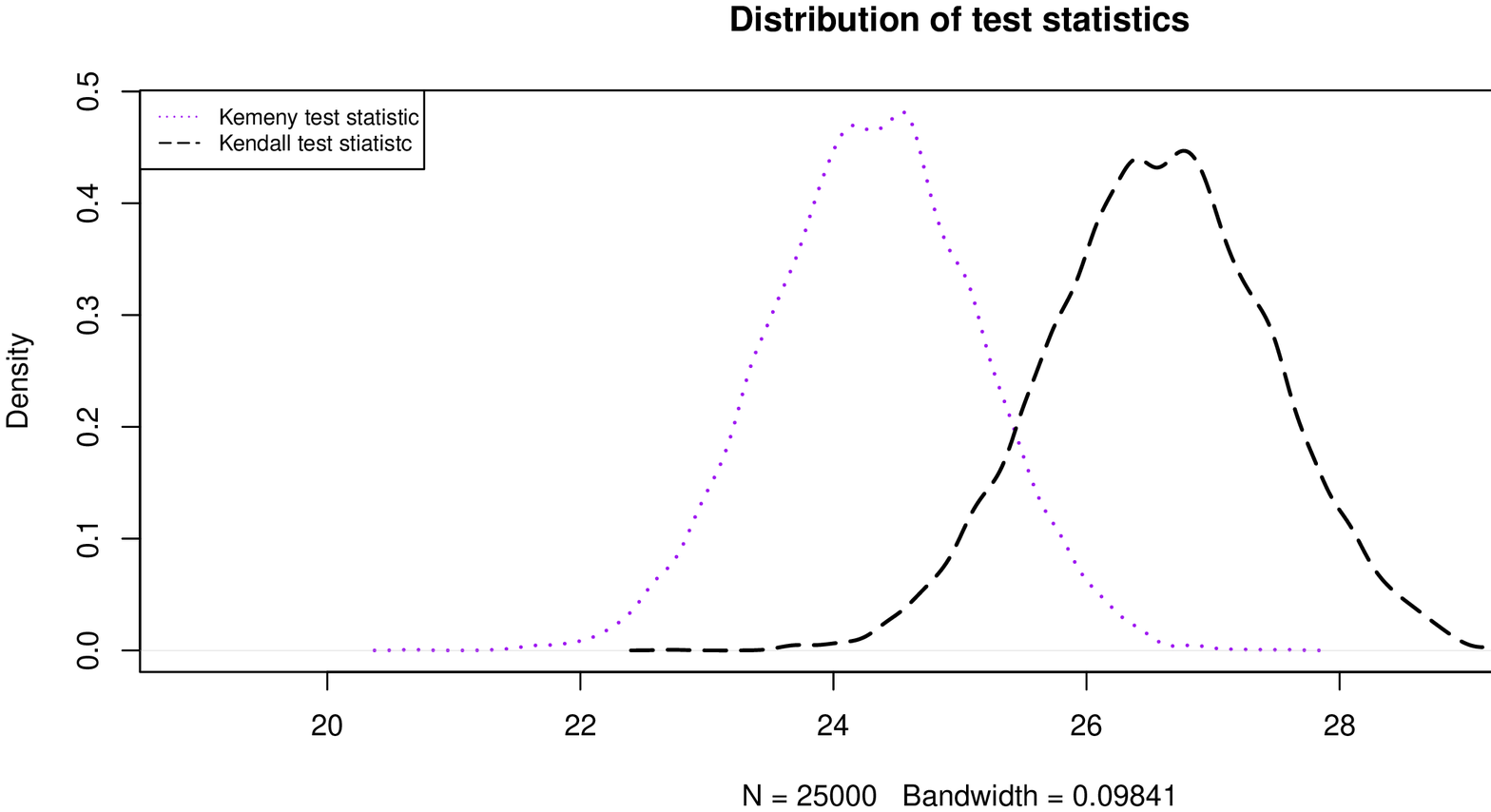}
% \end{subfigure}
\end{figure}

\begin{table}[!ht]
\centering
\scriptsize
\caption{Empirical comparison of the distribution of 25,000 test statistics for the bivariate correlations and Welch's \(t_{1248}-\)test over ordinal random variables for various estimators.}
\label{tab:3}
\begin{tabular}{cccccccccc}
  \toprule
Sample-size & & mean & sd & median & min & max & range & skew & kurtosis \\
  \midrule
\multirow{3}{*}{(a) n = 1250} & Kemeny \(t-\)test & 15.1668 & 0.8943 & 15.1709 & 11.3614 & 18.4783 & 7.1169 & -0.0392 & -0.0008 \\
& Kendall \(z-\)test & 16.3493 & 0.9521 & 16.3554 & 12.2476 & 19.8721 & 7.6245 & -0.0436 & 0.0056 \\
  & Student \(t_{1248}\tfrac{1}{c^{3}}\) & 15.1746 & 0.9169 & 15.1709 & 10.9595 & 18.9556 & 7.9961 & 0.0294 & -0.0014 \\
\midrule
\multirow{2}{*}{(b) n = 2500} & Kemeny \(t-\)test & 24.3124 & 0.8286 & 24.3218 & 20.6684 & 27.5784 & 6 .9101 & -0.0324 & -0.0012 \\
& Kendall \(z-\)test & 26.5708 & 0.8893 & 26.5763 & 22.7169 & 30.0486 & 7.3318 & -0.0298 & -0.0074 \\
  \bottomrule
\end{tabular}
\end{table}

\subsubsection{Paired sample non-parametric t-test}
Assume the existence of a bivariate population, wherein independent random variable \(x\) is binomial and the ground state occurs with probability \(1 - \pi\), and that \(y\) is a random variable upon which the elements \(y_{i}, i = 1,\ldots,n\) may be ordered. Each random variable may be indexed by subject \(i\), and thus the two random variables denote a common measurement upon each individual under different conditions: thus is posed a paired t-test. The non-parametric equivalent is the Wilcoxon signed rank test. Consider the Sleep data set once again, wherein \(n = 10\) observations are made between two groups of the same individuals, treated with two distinct drug assignments. Using the following linear expectation, 
\begin{equation}
 -\frac{\rho_{\kappa}(X,Y)}{\tfrac{\dot{\sigma}^{2}_{\kappa}(n)}{\hat{\sigma}_{\kappa}(X-Y)}} \sim t_{n-1}
\end{equation}
we obtain the standard deviation of the variance adjustment to equation~\ref{eq:kem_population_variance} from equation~\ref{eq:kem_variance} upon the difference in scalar differences between the ordering of the first and second trials. In the event of greater dispersion of the ranks (e.g., more ties and therefore greater uncertainty with respect to a true ordering) the denominator increases, and therefore the consistent estimator (and corresponding test statistic)  to the Wilcoxon rank-sum test in the presence and absence of ties, becomes smaller, indicating less power relative to the numerator, which contains a constant signed distance in the interval of the neighbourhood \(U^{*}(\mathcal{M})\).  
\subsection{Duality between the Kemeny norm measure space and Frobenius norm measure space}

The concept of duality for projective geometry is a natural avenue of investigation, first explored for correlations by \textcite[p.~129]{kendall1948}. The concept holds that for the planar projective geometry of the Euclidean space, there exists a dual permutation geometry, one which has been observed here between the Kemeny and Euclidean metric spaces. In a formal sense, by the finite nature of the Kemeny metric space, we may consider it to be a Galois field, and further a dual vector space, satisfying the three necessary properties of a dual cone, which are both continuous spaces upon their respective measurement norms.

Further, the standard construction of the existence of a duality, the ability to distinguish between identical elements upon a given field with a second, is clearly self-evident upon the dual metric space characterisation. We ignore for the moment the limiting case of the linear permutation field upon a population of linear scores (i.e., the standard asymptotic parametric learning problem per \cite{le1986}). In such a scenario is observed the perfect parametric score fit upon the population, which implies an equivalence in perfect ordering, thus denoting a bijective relationship between the ranking and the scoring through the cumulative distribution function (CDF). Instead we demonstrate that, especially (albeit non-uniquely) in the problem of Tikhinov regularised, ill-posed or biased, learning upon linear functional map, the duality of the two metric spaces grants a just-identified unique solution. This characterisation is extended in later work to demonstrate a unique solution to the Karush-Kuhn-Tucker constrained optimisation (i.e., ridge regression) in a multitude of final sample problems common to the Social Sciences.

In \textcite[p.~129]{kendall1948} the following equivalence was claimed between Pearson's \(r\) and Kendall's \(\tau_{b}\): 
\begin{equation}
\label{eq:kendall_sin}
r_{x,y} = \sin\bigg({\tau_{b}(x,y)\cdot\frac{\pi}{2}}\bigg).
\end{equation}
We proceed to show this characterisation to be invalid, and instead demonstrate that the left-hand side of equation~\ref{eq:kendall_sin} is actually Kemeny's \(\rho\), constructable from \(\kappa(x)\) and \(\kappa(y)\), which only asymptotically converges to Pearson's \(r\) upon the population, thereby denying the existence of the strict bijection upon a sample. Our interest is in establishing that for a domain in which ties are present, the two correlation coefficients are not equal, and thus under no further assumptions can Pearson's \(r\) result from the invalid estimator upon any domain except that of \(S_{n}\), and then only when the bivariate distribution is further Gaussian and linearly well-posed. This is an exceptionally important realisation for when one extends the bivariate correlation to scenarios in which multiple predictors are simultaneously solved for in relation to a common target (i.e., standard multiple regression) upon non-Euclidean function spaces. 

First, consider the definition of the \(\ell_{2}\) or Frobenius norm -- the insertion of one or more infinite values explicitly results in the expectation of at least one of the random variables to be non-finite (and thus degenerate). Then by the non-finite expectation, the \(n \times 1\) vector upon the Euclidean metric space is neither capable of being centred, nor is the relative distance capable of being assessed. The inner-product is also undefined, as the sum of inner-product (i.e., the cross-product) non-finite values is itself non-finite, and thus the Pearson correlation measure is degenerate. This extension of the domain is necessary however, as the Kemeny correlation is validly applied upon the extended reals, and produces a finite measure concomitant for any finite \(n\) ordering upon a vector. This results in a paradox, wherein the Kemeny correlation exists and obtains a finite convergent value, while the Pearson correlation is degenerately non-measurable, in contradiction of equation~\ref{eq:kendall_sin}. Thus, the relationship defined in equation~\ref{eq:kendall_sin} is invalid for all scenarios is which a linear projection of the scores of between the variables is non-linear, which includes but is not extended to the extended real line. Moreover, for finite samples even under the weak law of large numbers, if the random variables are non-Gaussian, the Pearson correlation is biased by definition, and therefore cannot satisfy the equivalence.

However, we can construct from the Kemeny metric a substitute estimator for the Pearson correlation, which is easily found to be equivalent upon the population. For each extended real random variable \(\{x \equiv \kappa(x)^{n \times n }, y \equiv \kappa(y)^{n \times n}\}\), we desire a bivariate vector matrix of order \(n \times 2\) from the otherwise extant tensor \(n \times n \times 2\), to comply with the necessary inner-product formulation of the Pearson product-moment correlation. We obtain this upon the skew-symmetric permutation domain by taking the sum over all \(k\) rows for each skew-symmetric matrix, resulting in the production of two \(1 \times n\) vectors, each with expectation of 0, denoted as \(\vec{x}^{\intercal}_{n \times 1}\) and \(\vec{y}^{\intercal}_{n \times 1}\), respectively. This is produced as follows, substituting \(\vec{y}\) for \(\vec{x}\) as appropriate:
\begin{equation}
\label{eq:kemeny_rho}
\begin{aligned}
\vec{x}: \overline{\mathbb{R}}^{n \times 1} \to \kappa(x) \in \mathbb{R}^{n \times n} \to  (\mathbb{R}^{1 \times n})^{\intercal} \\
\vec{x} = \sum_{k=1}^{n} \kappa_{k}(x) = \bigg[ \sum_{k}^{n} \kappa_{l=1}(x) , \sum_{k}^{n}\kappa_{l=2}(x) , \cdots , \sum_{k}^{n}\kappa_{l=n}(x)\bigg]^{\intercal}. 
\end{aligned}
\end{equation}
We note the expectations, \(E(\vec{x}) = 0\) and \(E(\vec{y})= 0\) are fixed for all samples, and a population variance along with Studentising standard deviation adjustment concentration measure is uniquely estimated and identified for each univariate random variable (equation~\ref{eq:kem_variance}). Each \(\vec{x}\) is then a vector of the rank ordering of a variable, with finite mean and variance, even in the observation of a non-finite variate value in either \(x\) or \(y\), as a linear ordering is still a valid representative mapping: The inner-product of \(\langle\tfrac{\vec{x}}{\sigma_{\kappa}(\vec{x})},\tfrac{\vec{y}}{\sigma_{\kappa}(\vec{y})}\rangle\) is the cosine of the linear ordering upon the vectors, with fixed unit distance by the law of parallelograms upon a Hilbert space.

We must first though establish that this inner-product is valid upon the vector of the extended reals, unlike the Pearson correlation. Upon \(\mathrm{S}_{n}\), the space of permutations wherein ties occur with probability 0, we observe that the inner-product of the ranks of finite scores and the scores themselves are equivalent. This is also validly asserted by noting the previously established fact that Kendall's \(\tau\) is a strict subset of Kemeny's \(\tau_{\kappa}\) and the corresponding distances, and therefore any measure upon a common sub-domain must and does measure equivalently upon the other as well. However, the inner-product of the two spaces reflects the Euclidean distance between pairs of rankings, which under a bijective equivalence upon the rank and score, in turn would equate Pearson's \(r\) and Spearman's \(\rho\) by the weak law of large numbers. 
% : it is then only in the presence of ties in a common domain for which non-equivalent estimates of the cosine of the angle may follow. As Kendall explicitly ignored this case, the mistaken identification \(\vec{X} \equiv X\) and the respective inner-products is understandable, if still demonstrably inaccurate. It is also trivially understood that both \(\vec{X}\) and \(\vec{Y}\) are the rank vectors of the original corresponding variables, and it may be observed that under Kendall's sinusoidal relationship between the inner products holds here between equations~\ref{eq:kemeny_rho} and equation~\ref{eq:kem_cor}. This invalidates the explicit existence of a relationship between the rank and the score inner-products and instead we must replace the left hand side of equation~\ref{eq:kendall_sin} with the almost equivalent correlation upon the ranks (Kemeny's \(\rho\)). As previously noted in the limiting case however, Kemeny's \(\rho\) converges to Pearson's \(r\), now holding over a more complete multivariate function space (as adjusted to allow ties) between Kemeny \(\rho\) and Pearson's \(r\).  
With this paradox resolved between the relationship upon the divergent metric spaces, we return to the question of the efficiency of the bivariate rank-ordering of a distribution of uniformly sampled variates upon \(\overline{\mathbb{R}}^{n \times p}\). This allows for a linear functions to be employed over any linearly orderable bivariate distribution, while satisfying the Cr\'{a}mer-Rao lower-bound (Lemma~\ref{lem:cramer-rao}) in order to assess the sufficient statistics of any \(n \times p\) distribution. Of note, this resolves to the identification of the median and variance of each of \(p\) variates, along with the correlation matrix \(\Xi_{p \times p}\), as an \(\ell_{2}\)-norm space, thereby identifying a quadratic solution for the expectation.

% This solution is an explicit consequence of the parametric case, in which the permutation ring of the target Kemeny distance is minimised to 0 (i.e., the correct ordering of the projective dual, the predictions, mirrors the target exactly) at the same time that the error in scores are also exactly correct. As there are a valid set of permutations which are equally distant from the target but which have different scoring error, we obtain a dual measurement which is uniquely minimised when both the Euclidean and Kemeny distances are conjointly minimised.  

This duality allows for the convergent assessment of even a Tikhinov regularised system of linear equations to be uniquely identified and solved, as the existence of any viable solution is guaranteed to be in the neighbourhood of the true solution by Lemma~\ref{lem:lower_equality} and Lemma~\ref{lem:upper_equality}, which converges monotonically to the expectation upon the population with probability 1. Therefore, even when the Euclidean metric proposes an ill-posed (and therefore non-unique) solution, the conjunctive combination with the Kemeny metric and its projective geometric duality must allow for a unique solution which is both closest in observed score and true (unbiased) ordering, for all finite samples. Accepting then that \(\Pr(x\mid \alpha_{1},\alpha_{2})\propto x,\) and \(\Pr(y\mid \alpha_{1},\alpha_{2})\propto y,\), it then must follow by the law of total expectations that \(E(\Pr(\vec{x})\cdot \Pr(\vec{y})) = \Pr(\vec{x}\vec{y}).\) A likelihood function may then be constructed as a function of the marginal shape coefficients and the the inner-product or correlation coefficient, which is just-identified. The derivative of \(L(\rho_{\kappa}\mid \vec{x},\vec{y},\propto \alpha_{x,1},\alpha_{x,2},\alpha_{y,1},\alpha_{y,2}\)     

\section{Discussion}

We conclude with a discussion of the relative value for non-parametric linear model spaces. In traditional statistical literature, it is nearly universally found that the utility of the variate values is exceptionally important. This is trivial to demonstrate, as in the parametric estimation problem, we observe that the rank is a bijective function of the scores, and thus it presents an integral component for identification. However, in all biased or Tikhinov regularised linear model spaces, the uniquely defined inverse would be exceptionally useful. In fact, using the Kemeny correlation matrix as a basis from which to construct an EM-algorithm for addressing missing data, clustering, Factor Analysis, and general multiple regression problems have been conducted with immensely positive results. The fundamental problem has always been the lack of identification of a probability model in the presence of ties, and unsurprisingly, a surjective mapping of linear combinations almost surely results in such ties. Trivial examples may be noted in the presence of a non-parametric 2-way ANOVA analysis of dichotomous groups, in which ties result by construction, and is otherwise incapable of being estimated. This complication would suitably explain what we perceive to be the lack of development over the preceding 80 years, as we otherwise require substantive restrictions or binomial explosion conflated problems to be uniquely estimated upon each sample (i.e., the standard order statistic construction scenario). The ability to construct a complete Hilbert space for arbitrary orderable distributions, probabilistically, explicitly resolves, in our opinion, a foundational and necessary unresolved problem in current Statistical Literature. For instance, the combination of rank and score based approaches allows for many otherwise unsolvable ill-posed system of equations to now be uniquely solved. Further, by the noted duality of the two spaces, this structure is identical (as has been empirically confirmed) to what would otherwise be the unique solution upon the parametric case. 

The possibility for the development of unique solutions to general machine learning problems (which would otherwise satisfice the definition of a Hadamard ill-posed problem) as a dual defined loss function is of immense interest and development and also has intuitive appeal as well. In the parametric case, the scores intrinsically denote an ordering upon the CDF. The solution to the system of estimating equations which possesses both the ability to order all cases correctly and estimate the location and scale of said data items is the minimum variance unbiased estimator. If addressing a system of ill-posed estimating equations though, the set of solutions upon the Euclidean distance function space which also produces the closest approximation of the ordering of the scores relative to the target satisfies the Lagrange multiplier (i.e., saddle-point) solution to the most stable solution of all candidates, which is otherwise infinitely defined by the completely arbitrary choice of the Tikhinov bias. Thus, the projective geometric duality allows for unique, generalisable, solutions upon all available data without requiring any further assumptions, restrictions, and minimal computational burdens than were already imposed by the Euclidean distance loss function.  

This describes a resolution to the otherwise extant paradox introduced by the likelihood principle, wherein all information is contained by knowledge of the correct probability distribution. By the duality of the two function spaces, inclusion of two sets of correctly identified probability distributions is redundant, and therefore no new information upon the population distributional structure is gained by over-parametrisation. However, when the parametric distribution is only approximated, rather than exactly known, the inclusion of a second orthonormal function space allows for the bias introduced upon the Euclidean distance space to become just-identified. Conjointly then, the dual optimisation of the complementary measure spaces must increase the relative information using only one metric. Considering that the utility of an ordering without a scale and location is functionally meaningless for many different learning problems (e.g., the correct order distribution of \(n\) finite elements in the unit interval contains an infinite variate realisations which comply with the `correct' ordering), the natural starting place is instead that we wish to acquire the uniquely correct ordering which also possesses the most accurate empirical approximation of the Euclidean distance as well. Of course, given this duality, the parametric scenario makes the metric duality redundant asymptotically. Yet however, the uncertainty of the stochastic perturbations upon finite samples provides an allowance for information gain even in the presence of the correct probability distribution, conditional upon the uniformly arising random sampling.

\small
\printbibliography

@Book{le1986,
  Title                    = {Asymptotic methods in statistical decision theory},
  Author                   = {Le Cam, Lucien Marie},
  Year                     = {2012},
  Publisher                = {Springer Science \& Business Media}
}

@article{diaconis1977,
	doi = {10.1111/j.2517-6161.1977.tb01624.x},
	year = {1977},
	publisher = {Wiley},
	volume = {39},
	pages = {262--268},
	author = {Persi Diaconis and R. L. Graham},
	title = { Spearman's Footrule as a Measure of Disarray },
	journal = {Journal of the Royal Statistical Society: Series B (Methodological)}
}

@article{kemeny1959,
	doi = {10.2140/pjm.1959.9.1179},
	year = 1959,
	publisher = {Mathematical Sciences Publishers},
	volume = {9},
	number = {4},
	pages = {1179--1189},
	author = {John G Kemeny},
	title = {Generalized random variables},
	journal = {Pacific Journal of Mathematics}
}

@article{kendall1948,
  title={Rank correlation methods.},
  author={Kendall, Maurice George},
  year={1948},
  publisher={London, UK; Griffin}
}

@article{nelder1972,
  doi = {10.2307/2344614},
  year = {1972},
  publisher = {{JSTOR}},
  volume = {135},
  number = {3},
  pages = {370},
  author = {J. A. Nelder and R. W. M. Wedderburn},
  title = {Generalized Linear Models},
  journal = {Journal of the Royal Statistical Society. Series A (General)}
}

@article{schoenberg1938,
  doi = {10.1090/s0002-9947-1938-1501980-0},
  url = {https://doi.org/10.1090/s0002-9947-1938-1501980-0},
  year = {1938},
  publisher = {American Mathematical Society ({AMS})},
  volume = {44},
  number = {3},
  pages = {522--536},
  author = {I. J. Schoenberg},
  title = {Metric spaces and positive definite functions},
  journal = {Transactions of the American Mathematical Society}
}

@book{bogachev2007,
  doi = {10.1007/978-3-540-34514-5},
  url = {https://doi.org/10.1007/978-3-540-34514-5},
  year = {2007},
  publisher = {Springer Berlin Heidelberg},
  author = {Vladimir I. Bogachev},
  title = {Measure Theory}
}

@book{buldygin2000,
  title     = {Metric characterization of random variables and random processes},
  author    = {Buldygin, Valery and Kozachenko, Yu V},
  publisher = {American Mathematical Society},
  series    = {Translations of mathematical monographs},
  month     =  {feb},
  year      =  {2000},
  vol = {188},
  address   = {Providence, RI},
}

@article{semrl1996,
title = {On a matrix version of Cochran's statistical theorem},
journal = {Linear Algebra and its Applications},
volume = {237-238},
pages = {477-487},
year = {1996},
note = {Linear Algebra and Statistics: In Celebration of C. R. Rao's 75th Birthday (September 10, 1995)},
issn = {0024-3795},
doi = {https://doi.org/10.1016/0024-3795(95)00395-9},
url = {https://www.sciencedirect.com/science/article/pii/0024379595003959},
author = {Peter Semrl}
}

@book{hollander2014,
   title =     {Nonparametric Statistical Methods},
   author =    {Myles Hollander and Douglas A. Wolfe and Eric Chicken},
   publisher = {John Wiley \& Sons},
   isbn =      {9780470387375},
   year =      {2014},
   series =    {Wiley Series in Probability and Statistics},
   edition =   {3rd}
}
% \clearpage
\appendix
\section{Appendix}
\begin{theorem}
\label{lem:hilbert}
% From Definition~\ref{def:hilbert}, a Banach norm space such as
The Kemeny norm space is a pre-Hilbert space with the property that for finite positive \(a\), holds the following equivalence, using \(\kappa^{*}\) to denote the signed function as given in \textcite{kemeny1959}:
% shown to be a pre-Hilbert metric space if for $a > 0$ follows the equivalence denoted by positive homogeneity such that
\[\rho_{\kappa}(a\cdot\kappa^{*}(X_{A}),\kappa^{*}(X_{B})) = a\rho_{\kappa}(\kappa^{*}(X_{A}),\kappa^{*}(X_{B})).\]
\end{theorem}
\begin{proof}
By scalar multiplication, holds the lower bound of $0\cdot{a} =0$, from Lemma~\ref{lem:lower_bound}, and the upper bound (Lemma~\ref{lem:cauchy}) exists with a real number for all finite $n$. The canonical form of the Kemeny distance function is therefore scaled by $a^{-1}a = 1$, and therefore the linear distances are always equivalent to the closed interval $[0,1,\cdots,(n^{2}-n)], \forall\, n\ge 0$ if \(a=1\), and which is otherwise proportional to this sequence by positive finite scalar \(a \in \{(0,\infty^{+})\}\setminus 1\). Thus it is shown that the Kemeny Banach norm-space is also positive homogeneous while lacking an inner-product construction, and is therefore a pre-Hilbert space.
\end{proof}
\begin{corollary}
From Theorem~\ref{lem:hilbert} we have justified the valid existence of an inner-product formulation for the Kemeny pre-Hilbert metric space, by definition, which is provided in equation~\ref{eq:kem_dist}.
\end{corollary}

\begin{lemma}~\label{lem:lower_bound}
The Kemeny metric space is non-negative under the condition that \(0 < a < \infty^{+}\).
% Non-negativity of the Kemeny metric space, under certain conditions, follows from the axiomatic use of the metric properties (Definition~\ref{def:metric_space}) as proven by \cite{kemeny1959}.
When the domain of a univariate random variable in \(\mathbb{R}^{n \times 1}\) is extended to a bivariate pair \((X_{A},X_{B})\), the function \(\kappa(X)\) is convex, for the arbitrary arbitrary real constant \(\alpha\).
\end{lemma}
\begin{proof}
By sequential application of sub-additivity, symmetry, and the identity of indiscernibles follows non-negativity, but only under a specific definition of $\kappa(\cdot)$, upon which results the requirement $0 < a < \infty^{+}$:
\begin{equation}
\begin{aligned}
% \rho_{\kappa^{*}}(\kappa^{*}(X_{A}),\kappa^{*}(X_{B})) + \rho_{\kappa^{*}}(\kappa^{*}(X_{B}),\kappa^{*}(X_{A})) & \ge \rho_{\kappa^{*}}(\kappa^{*}(X_{A}),\kappa^{*}(X_{A}))\\
% \rho_{\kappa^{*}}(\kappa^{*}(X_{A}),\kappa^{*}(X_{B})) + \rho_{\kappa^{*}}(\kappa^{*}(x),\kappa^{*}(X_{B})) & \geq \rho_{\kappa^{*}}(\kappa^{*}(X_{A}),\kappa^{*}(X_{A}))  \\
% 2\cdot \rho_{\kappa^{*}}(\kappa^{*}(X_{A}),\kappa^{*}(X_{B})) & \ge 0\\
\rho_{\kappa^{*}}(a,b) + \rho_{\kappa^{*}}(b,a) & \ge \rho_{\kappa^{*}}(a,a)\\
\rho_{\kappa^{*}}(a,b) + \rho_{\kappa^{*}}(a,b) & \geq \rho_{\kappa^{*}}(a,b)  \\
2\cdot \rho_{\kappa^{*}}(a,b) & \equiv 2\cdot \rho_{\kappa^{*}}(\kappa^{*}(X_{A},\kappa^{*}(X_{B})) \ge 0\\
\rho_{\kappa^{*}}(\kappa^{*}(X_{A},\kappa^{*}(X_{B})) \ge 0\\
\end{aligned}
\end{equation}
\end{proof}

\begin{definition}[Complete space]
\label{def:complete_space}
A metric space \((X,\rho)\) is complete if the expansion constant of the metric space is \(\le 2.\)
\end{definition}

\begin{lemma}~\label{lem:kem_expansion}
% An alternative proof of the completeness of the Kemeny metric space is also offered by the expansion constant of the Kemeny distance function, from Definition~\ref{def:expansion_constant}, which is here shown to be less than 2 for the Kemeny metric for finite \(n\):
The Kemeny distance function has an expansion constant \(\mu\) less than 2.
\end{lemma}
\begin{proof}
Let there exist two families $\{A_{i}\}_{i \in I}$ and $\{B_{j}\}_{j \in J\,, i,j = 1,\ldots,n}$ of $n$ real numbers such that $a_{i} \le b_{j}$ and therefore 
\begin{equation}
\begin{aligned}
\rho_{\kappa^{*}}(\kappa^{*}(A_{i}),\kappa^{*}(B_{j}))\, \forall\, i,j \implies & \sup_{i\in I}(\kappa^{*}(A_{i}))  \equiv \mu \\
& = \sup\{-1\cdot{a},0\cdot{a}\} \le \inf_{j\in J}(\kappa^{*}(B_{j})) \equiv \inf\{0\cdot{a},1\cdot{a}\}\\
& = \sup\{\mu\} \le 1.
\end{aligned}
\end{equation}
It therefore follows that for any family of intervals $A_{i},B_{j}$ whose intersection of the two finite sets is never empty, $\cap[\sup{A_{i}},\inf{B_{j}}] \ne \varnothing$, possesses an expansion constant $\mu = \sup\{0,1a\} \le 2a\; \forall a>0$. Consequently, Definition~\ref{def:complete_space} is satisfied for the Kemeny distance function, in that the isometric expansion constant is always less than 2; as \(X_{A} = A_{i},X_{B}=B_{j}\), this proof applies to all score vectors within \(\mathcal{X}\), the space of all vectors of extended reals of finite length \(n\).
\end{proof}

\begin{lemma}~\label{lem:cauchy}
% Next, we prove the existence of a unique finite upper-bound upon the Kemeny metric space for any collection of $n$ elements, whose lower bound is known to be 0 by the axiomatic properties of a metric space (Definition~\ref{def:metric_space}) and by Lemma~\ref{lem:lower_bound}.
There exists a unique finite upper-bound upon the Kemeny metric space for any collection of \(n\) element length vectors in the extended reals, whose lower bound is 0.
\end{lemma}
\begin{proof}
For any $\kappa^{*}(X_{j})\; \forall X_{j}, j = 1,\ldots,p$ results a mapping vector of maximum $\frac{(n^{2}-n)}{2}$, to which are assigned one of three distinct values: $a_{ij} \in \{1,-1,0\}$. If all $X_{i},\, i = 1,\ldots,n$ are in a monotonically ascending sequence, then there is \(\mathcal{M} = n^{n}-n, \forall 0 < n < \infty^{+}\), representing all vectors of extended real scores, $n^{n} = \{n_{1}, n_{2}, \cdots, n_{n}\}$. Each vector, under operation $\kappa^{*}$ results a mapping of dimension $\frac{(n^{2}-n)}{2}$; the reduced set of mappings is granted by the symmetry of the metric space. It therefore follows that for fixed $n$ upon monotonic sequence \(X_{j}\) exists $ I_{n} = \{1,2,3,\ldots\} \in X_{n}$, form which results $\sum_{i=1}^{n}\kappa^{*}(X_{j}) = \sum \{1a,1a,\ldots,1a\}$. By equation~\ref{eq:kem_dist} it then follows that upon such a field, the maximum distance attainable is found by the sequence of the absolute value of the difference for a similarly monotonically ordered sequence, for which is substituted the value  $1a \to -1a$, with exception only upon \(a_{i=j} = 0\): \[
\sum_{i=1}^{n}\sum_{j=1}^{n} |a^{*}_{ij} - (-1\cdot a^{*}_{ij})| = |\sum_{i=1}^{n} \{2_{i}a,\ldots,2_{n}a\}| = \frac{n^{2}-n}{2} \cdot 2a = a(n^{2}-n).\] Uniqueness proceeds as with the preceding Lemma, for which the summation over all mapping \(\kappa^{*}\) with one or more $a$ are by definition of equation~\ref{eq:kem_dist} of distance greater than the observed, and therefore is uniquely extremised for any fixed finite positive $a$ and finite positive integer \(n\).%, and is therefore convergent within the domain of all extended real vectors within \([0,a(n^{2}-n)]\).
\end{proof}
\begin{lemma}
\label{lem:complete_space}
% Definition~\ref{def:complete_space} holds that a sequence in a metric space is Cauchy if for every positive real number $r > 0$ there is a positive integer $N$ such that for all positive integer vectors $X_{A},X_{B} \ge N$, $\rho(X_{A},X_{B}) < r.$
The Kemeny metric space is a complete metric space.
\end{lemma}
\begin{proof}
By the validity of both Lemma~\ref{lem:lower_bound} and Lemma~\ref{lem:cauchy} for the Kemeny metric, the maximum possible distance for said distance function, and any finite sequence $X$ of length $n$ is $a(n^{2}-n)$, as resulting from the additive sequence of $\frac{n^{2}-n}{2}$ elements, each valued $2a$, thereby establishing the definition off a positive and finite $r$ for all finite \(n\). Cauchy convergence therefore follows from the constructive mapping that for every $X_{A},X_{B} \in \mathcal{M} \setminus \{I_{n},I_{N}^{\prime}\}$, there exists a finite distance $0 \le |\rho_{\kappa}(\kappa^{*}(X_{A}) - \kappa^{*}(X_{B}))| < a\cdot{r} < a(n^{2}-n)$ and it is shown that the Kemeny metric space is Cauchy convergent, per Definition~\ref{def:complete_space}.  An immediate corollary follows such that the upper and lower bounds of the Kemeny distance are determined by the conjunctive choice of finite $n$ and $a$, s.t., \[0\cdot{a} \le \rho_{\kappa}(a\kappa^{*}(X_{A}),\kappa^{*}(X_{B})) \le a(n^{2}-n),\] and the lemma is concluded, as the Kemeny metric space is shown to be complete for \(X_{n}\in \mathcal{X}\).
\end{proof}

\begin{definition}[Markov inequality]~\label{def:markov_inequality}
The Markov inequality is defined to random variables within $M$ which take upon only non-negative values: this definition is concomitant with the definition of a complete metric function, and therefore is satisfied for all $n$ from which is defined the sub-space $M$. The Markov inequality holds that upon the exhaustively observed population, the random variable $\rho_{\kappa}(x_{n} )$ is such that for any $\epsilon>0$, the probability that the observed realisation is no greater than a distance $\epsilon$ from the variable is defined by the inequality ratio \[\text{Pr}(x \ge \epsilon) \le \frac{E(\rho_{\kappa}(x_{n}))}{\epsilon}.\] To aid in expressing this inequality, let us define $x \in M$, such that relative to the origin permutation $I_{n} = 1,2,\ldots,n$, there are $M$ distinct permutations with distances which are guaranteed to be in the closed interval $[0a,\ldots,a(n^{2}-n)]$ with positive probability, and to otherwise occur with probability 0.
\end{definition}

\begin{lemma}[Markov inequality]
\label{lem:markov}
Let $X$ be a random variable with a probability density function $f_{X}(\cdot)$ which arises as a consequence of the Kemeny distribution, cumulative distribution function $F_{\kappa(X)}(\cdot)$, and let $(X_{1},X_{2},\cdots,X_{N})$ be a random sequence, or sample, drawn upon said distribution. By Definition~\ref{def:markov_inequality}, it is observed that for all real numbers $\rho_{\kappa}(x_{n}) \ge \epsilon$ upon which the first moment-expectation $\mu_{\kappa} \in \mathbb{R}$ exists, Markov's inequality holds in expectation,
\begin{equation}
P(\rho_{\kappa}(x_{n}) \ge \epsilon) \le \min \big[\frac{\mu_{\kappa}}{\epsilon},a\big],
\end{equation}
and for which a strong upper-bound upon the distance holds such that the absolute distance between the expectation and any other observable point, for known $(n,a)$ as in equation~\ref{eq:kem_cor}, is never greater than $(0\le \epsilon \le a(n^{2}-n)$ (Lemma~\ref{lem:complete_space}). 
\end{lemma}
\begin{proof}
Then follows
\begin{align*}
E(x_{n}) & = \sum_{i=\infty^{-}}^{\infty^{+}} x_{i} p(x_{i}), s.t., \sum_{i=\infty^{-}}^{\infty^{+}} p(x_{i}) = 1\\
E(x) & = \sum_{i=0}^{\epsilon} x_{i}p(x_{i}) + \sum_{i>\epsilon}^{a(n^{2}-n)} x_{i} p(x_{i}) \ge \sum_{i>\epsilon}^{a(n^{2}-n)} \epsilon \cdot \big(x_{i} p(x_{i})\big) \\
E(x) & = \epsilon \sum_{i>\epsilon}^{a(n^{2}-n)} p(x_{i}) = \epsilon P(|x| \ge E(x) \ge \epsilon)\\
\frac{E(x)}{\epsilon} & = P(x \ge \epsilon).
\end{align*}
\end{proof}

\begin{lemma}[Tchebyshef inequality]~\label{lem:chebyshev}
The first expectation, expressed $E(\rho_{\kappa}(x)) = \mu_{\kappa} = a(\frac{n^{2}-n}{2})$, and the variance of $\rho_{\kappa}(x)$ is expressed as $\sigma^{2}_{\kappa}$, for an arbitrary non-negative sequence of $t$. Then,
\begin{equation}
P(|\rho_{\kappa} - \mu_{\kappa}| \ge t) \le \frac{\sigma_{\kappa}^{2}}{t^{2}},
\end{equation}
by Markov's inequality, as \[P(|\rho_{\kappa}(x_{n} - \mu_{\kappa}) \ge t) = P((\rho_{\kappa}(x_{n}) - \mu_{\kappa})^{2} \ge t^{2}) \le \frac{ E\big( (\rho_{\kappa}(x_{n}) - \mu_{\kappa})^{2}\big)}{t^{2}} = \frac{\sigma^{2}_{\kappa}}{t^{2}}.\]
\end{lemma}
\begin{proof}
This is constructively simple to understand, and complies directly with the expected behaviour of a finite range space for the Kemeny metric under evaluation: that no more than a certain real fraction of values are to be found to be greater than a distance of the real $t$ standard deviations away, denoting symmetrical coverage bounded in probability by $1-\frac{1}{t^{2}}$. 

By the finite distance as a linear function of $n$ is strictly known to be no more than a real finite distance $\frac{1}{2}\big(n^{2}-n\big)$, which is the symmetric maximum distance, $\max\epsilon$ from the expectation, it is thereby demonstrated that the upper and lower tail bounds of the Kemeny metric, for any collection of independently observed $n$ reals arising from a common distribution, thereby satisfying Tchebyshef's inequality.  
\end{proof}

\begin{definition}[Hilbert space]\label{def:hilbert}
A Hilbert space is a complete metric space (a Banach norm-space) on
which there is an inner product $\langle x, y \rangle$ for which holds
the following properties:
\begin{enumerate}
    \item{The inner product is conjugate symmetric and thus the inner
product of a pair of elements is equal to the complex conjugate of the
inner product of the swapped elements:
    \[\langle y,x\rangle ={\overline {\langle x,y\rangle }}\,.\]}
    \item{Said inner product is linear in its first (real) argument,
such that upon all complex numbers $a$ and $b$, holds \[\langle ax_{1}
+ bx_{2} , y \rangle = a \langle x_{ 1} , y \rangle + b \langle
x_{2},y\rangle .\]}
    \item{The inner product of an element with itself is positive definite}
\end{enumerate}
The existence of the property of positive homogeneity upon a Banach norm-space establishes the existence of an inner-product. From these properties follows the existence of the Cauchy-Schwarz inequality upon the metric $\rho$ s.t., \[\left|\langle x,y\rangle \right|\leq \|x\|\|y\|\] with equality \textit{iff} $x$ and $y$ are linearly dependent. A Hilbert space is therefore a complete metric space (Banach space) which possesses the further condition that it is positive conjugate homogeneous: \[f(cx) = c\cdot f(x),\ c > 0.\]
\end{definition}

\begin{lemma}[Even function]~\label{lem:even}
The $\kappa^{2}$ function is an even function, as are all constructed affine linear transformations thereupon.
\begin{definition}[Even and Odd Functions]

A function is even if for every input $x$, \(\rho_{\kappa}^{2}\left(x\right)=\rho_{\kappa}^{2}\left(-x\right)\). A function is correspondingly odd if for every input $x$ \(f\left(x\right)=-\rho_{\kappa}^{2}\left(-x\right)\). If a function satisfies \(\rho_{\kappa}^{2}\left(x\right)=\rho_{\kappa}^{2}\left(-x\right)\), it is even. Correspondingly, is a function satisfies \(\rho_{\kappa}^{2}\left(x\right)=-\rho_{\kappa}^{2}\left(-x\right)\), it is odd. If the function does not satisfy either rule, it is neither even nor odd.
\end{definition}
\begin{proof}
Consider $\kappa(\cdot)$ function as applied to a real vector of length $n$, x = \{1,\ldots,n\}. For such a mapping $\kappa(x)$ is produced a square matrix of order $n \times n$. By the constructive definition of such matrix (equation~\ref{eq:kem_score}) we obtain four conditions, of which three are unique: $\{x_{i} \le x_{j}\to -1, x_{i} > x_{j} \to 1, x_{i} = x_{j} \to 0\}$ upon which all $n$ elements results. For condition 3, observe that $0 = -1(0)$, and therefore comparisons of identical elements upon $x$ are neither even or odd. For the other conditions, $\kappa(x) = -1\kappa(x)^{\intercal}$, and therefore the $\kappa(\cdot)$ function in isolation is odd. However, $\kappa^{2}(x) = \kappa^{2}(-x)$ as the squared elements are therefore either $\{0,1\}$, and as no 0 is found unless the function is odd, and no elements may be negative by the squared transformation of any real, the affine linear function of any two pair of $\kappa$ functions is everywhere even.  
\end{proof}
\end{lemma}

\begin{lemma}\label{lem:sigma_lower}
We next proceed to prove that the second moment, the variance, upon the Kemeny measurement space is always finite subject solely to the assumption of a collection of independent observations. \end{lemma}
\begin{proof}
For $x \in X^{n \times 1}$, observe that $\kappa(x)\, \forall\, x \in \mathcal{M}$ is always a square matrix. The distance for any given $x$ in the domain of any metric to itself is observed to be always 0, by the definition of a metric space. However, the second moment of the variance is defined upon the $\kappa$, not its cross-product, and therefore the elements denote, effectively, the tabulated rate of distinctiveness (uniqueness) upon an real vector of length $n$. Therefore, let the polynomial function of $\kappa^{r}(\cdot)$ denote the $r^{th}$ moment of the vector $x$, where the first moment has already been shown to be 0, and therefore the raw moments and the central moments upon the Kemeny metric are obtained by the central limit theorem. Thus, it logically follows that there exist a finite upper and lower bound upon the Kemeny variance, as expected for any totally bounded space, such as the Kemeny metric. 

We first establish for $\kappa(x)$ the existence of a lower bound, which will be shown to be 0. This definition of the second moment is equivalent to a finite real constant vector, for which all $n$ elements are identical. Thus, there will be observed no such variability in the realised scores upon $x$. For any $\kappa(x) \in \mathcal{M}$ then, there exists a vector marginalised over $i,j = 1,\ldots,n$ \[  \sum_{j=1}^{n}\kappa_{ij}(x) = \sum_{i=1}^{n}\kappa_{i}(x) = 0 \equiv \sum_{i=1}^{n} -1 \sum_{j=1}^{n} \kappa_{ij}(x)\] There exist $n$ vectors for which all elements in $x$ are identical, and therefore all mappings of any $n$-tuple is a vector of $n$ 0's, since by condition three of the $\kappa(\cdot)$ function, all $n$ pairwise comparisons consist of the mappings for which $(x_{1},\cdots, x_{n}) \rightarrow a\cdot(0_{1},\cdots,0_{n})^{\intercal}$. For any $n$ therefore, it is seen that the square of such a sequence is always 0, for the general expression of the $r=2$ power of the $\kappa(\cdot)$ from which the second moment may be expressed 
\begin{equation}
\label{eq:kemeny_var}
\sigma_{\kappa}^{2}(x)=\sum_{i=1}^{n}\sum_{j=1}^{n} \kappa_{ij}^{2}(x) \ge 0.
\end{equation}
\end{proof}

\begin{lemma}~\label{lem:sigma_upper}
We next proceed to prove that the upper-bound, for any $n$ is also finite and must sum to a constant real which will be shown to be proportionate to a scalar function of $n$, such that the upper bound of the variance is defined to be always in the interval ring $\sigma_{\kappa}^{2} \propto [0,1]$. It will also be shown that when $\sigma_{\kappa}^{2}$ is maximised, the sub-space of the domain is equivalent to that of the traditional $S_{n}$ group from which the original order-statistics concept arises.  
\end{lemma}
\begin{proof}
First, note that for all \(\mathcal{M}\) permutations upon the vector $x$ of length $n$, there are \(\mathcal{M}\) $\kappa_{i}$ mappings, each of which possess an expectation of 0. The maximum variance is to be expected as the duplication of real numbers decreases; this is seen in the definition of the variance of 0 as a constant vector of real scores, such that all $n$ elements possess the same value. As the diversity of observable scores increases, the variance must also increases. Therefore, the maximum variance upon the orderings of a vector $x$ is to be observed when there are no $\kappa(x) = 0$ off of the diagonal in the skew-symmetric matrix and therefore a real magnitude is always equal to itself in the mapping resulting from the $\kappa$ function. 

Assume said skew-symmetric \(\kappa\) matrix, the expectation of the square of such a matrix with \(2(n**2-n)\) positive values whose maxima $\kappa_{i}^{2}(x)$ is no greater than $a^{2}(n^{2}-n)$. Such a measurement would require to be observed a vector $x$ for which an element in $x$ must satisfy either of two conditions: (1) $(a_{ij})^{\intercal}\cdot a_{ij} \equiv -(a_{ij})^{\intercal} a_{ji}$ or (2) $(a_{ij})^{\intercal}\cdot a_{ij} \ne -(a_{ij})^{\intercal} a_{ji}.$ We observe first that the only scenario for which the sign of $\pm a$, the product of two elements in a vector much always produce 0, by equation~\ref{eq:kem_score} and therefore a tie, occurs. For all other elements in the mapping $ x \to \kappa(x)$ then, the skew-symmetric nature of the \(\kappa\) matrix enforces that for the sum of $\sqrt{.5} \cdot n$ elements to be greater than $a(n^{2}-n)$. This would require $-(a_{ij})^{\intercal}\cdot(a_{ij}) = -(a_{ij})^{\intercal}\cdot(a_{ij})$ and therefore for $(-a) = a, a \in \{-a,0,a\} \setminus{0} = \varnothing$. This is an empty set of solutions, and therefore is impossible to occur: thus it is proven that the maximum bound for $0 \le \sigma_{\kappa}^{2} \le a(n^{2} - n)$ for $a > 0, n \ge 0$. Allowing \(a = \sqrt{.5}\) then, we obtain a maximal variance on the support of \(\tfrac{n^{2}-n}{\sqrt{.5}^{2}}\), and thus, for all $n$ it is therefore observed that by the characterisation of the Kemeny distance, there exists no $n$ collection of real $x$ vectors for which there does not exist a finite variance no less than 0 and no greater than $\tfrac{1}{2}(n^{2} - n)$.
\end{proof} 

\begin{lemma}~\label{lem:partition}
Let \(F_{\kappa}\) be a distribution function for the Kemeny metric of a random variable realised upon \(\overline{\mathbb{R}}\). For each \(\epsilon>0\) there exists a finite partition of the extended real line such that for an orderable sequence \(\infty^{-} \le t_{1} \le \cdots \le t_{k} \le \infty^{+}\) by Lemma~\ref{lem:lsn}, and  there exists \(0 \le j \le k-1\)
\[F_{\kappa}(t_{j+1})^{-} -  F_{\kappa}(t_{j}) \le \epsilon.\]
\end{lemma}
\begin{proof}
Let \(0 < \epsilon\) be given, such that there exists monotone convergence. Allow \(t_{0} = \inf{\overline{\mathbb{R}}}\), for which \(j \ge 0\) we define \(t_{j+1} = \sup\{z: F(z) \le F(t_{j}) + \epsilon\}.\) Then by right continuity, there are a finite sequence of steps for which this definition is discontinuous, and we observe that for our definition of the \(\kappa\) function, this scenario does not occur upon any countable finite population. Thus is defined a transition state of monotonically decreasing distance sequences from the expectation upon \(F\), and thus a finite distance for any finite sample from the compact and totally bounded supremum Kemeny distance \(n^{2}-n\).
\end{proof}

\begin{lemma}
\label{lem:cramer-rao}
The Kemeny estimator functions satisfy the Cram\`{e}r-Rao lower bound upon the population \(\mathcal{M}\) constructed of asymptotic limit on \(n\).
\end{lemma}
\begin{proof}
The unbiasedness of the estimator function is established in Lemma~\ref{lem:unbiased}, and said estimator function observed to be asymptotically normally distributed as well by Lemma~\ref{lem:kem_asym_normal}. As the variance of the estimator function is strictly sub-Gaussian for all finite \(n\), the variance is a scalar constant ratio which converges to 1 as a linear function of all data distributions. Under these conditions, it follow that by the Gauss-Markov theorem (Theorem~\ref{thm:gauss-markov}) the asymptotic variance grows approach the above to the asymptotic variance of the normal distribution, which in the limit wrt \(n\) is 0, and therefore concludes the proof in obtaining the Cram\`{e}r-Rao lower bound \[\lim_{n\to\infty^{+}} \frac{n^{n}-n}{n^{n}}\sigma^{2}_{\kappa} \to \sigma^{2},\] 
as is necessary by the central limit theorem. The convergence is guaranteed by Lemma~\ref{lem:lower_equality}, which holds that for any random variable upon the Kemeny metric, the distance will converge to 0, which is isometric to the solution under the isometric Euclidean distance. The upper-bound upon the measurability ensures that even if convergence does not tend to 0 (as would occur when a model is incorrectly specified, and the conditional Bayes error rate upon the sample is greater than 0) then the error is still always uniquely identified. 

These properties demonstrate that the duality construction of the two estimators are complementary, rather than mutually exclusive, and thus the examination of both metric spaces upon finite samples must, also by definition, decrease the entropy, thereby increasing the information gained towards the mutually agreeable optima by upon both function spaces.
\end{proof}
% \begin{definition}[Inverse Binomial Transformation]
% \label{def:inverse_binom}
% Expresses a moment $\mu_{n}$ of a probability function $P(x)$ taken about 0, s.t.,
% \begin{subequations}
% \begin{gather*}
% \mu_{n}^{\prime}  = \langle x^{n}\rangle =  \sum_{i=0}^{n} x^{n}P(x).
% \end{gather*}
% \end{definition}

\begin{lemma}
\label{lem:lower_equality}
The total variational distance between any Euclidean linear function and an unbiased Kemeny linear function is lower-bounded by 0.
\end{lemma}
\begin{proof}
Let the sum of all finite collections of random variables from both sub-Gaussian and Gaussian fields be stable, allowing \(Z\) to be a random set of realised variates of length \(n\), and let \(X\) be an independent random variable with distribution function \(F_{\kappa}\) and for which \(N \in U(\mathcal{M})\) is any finite number defining the compact and totally bounded support of \(F_{\kappa}\). It would then follow that the tail probability of \(X\) beyond \(N\) is the chance that \(X \notin U(\mathcal{M})\). If \(F_{\kappa}\) is the Gaussian distribution then, the probability that \(N \ni F\) and thus that \(N\) is not measurable upon \(F_{\kappa}\), is 0: \[e_{G} = \Pr(|X| > 0,\dots,N) = \Pr(X \notin U(\mathcal{M}) = F(N) - \lim_{\varepsilon \to 0^{+}} F(-N-\varepsilon).\]

Now consider the finite Galois field \(\mathcal{G}\) which contains \(G\), and consists of the Beta-Binomial distribution which is also measured upon random variable \(X\). The total variation distance between \(Y = G(Z)\) and \(X = F(Z)\) is the total variation distance between their probability distributions, where \(\mathcal{F}\) is the (probability measure) Borel set sigma-algebra upon \(F\) and \(\mathcal{G}\) is the corresponding set upon \(G\), from which follows: \[\|X-Y\|_{TV} = \sup_{A\subset \mathcal{F}}\left| \Pr(X\in A) - \Pr(Y \in A) \right|.\] The Markov (Lemma~\ref{lem:markov}) and Tchebyshef (Lemma~\ref{lem:chebyshev}) inequalities are immediately seen to hold for any sub-Gaussian function measured upon the Kemeny distance function. This is due to possessing finite expectations (compact and totally bounded, Lemma~\ref{lem:kem_bounded}), as the expectation of \(|X|\) is bounded above by the integral of the common population, which is finite (Definition~\ref{def:stict_sg}) and therefore are always measurable for any random variable which arises. 

Assume \(N\) is large enough to make \(\Pr(X \in U(\mathcal{M}))\) trivially non-zero, satisfying all conditions for which \(n >2\). Truncating \(X\) at a rational fraction of \(N\) thereby removes all chance that it exceeds \(N\), for which the value of the new distribution function at \(G(x)\) is \(0\) for \(x < -N\), \(1\) for \(x \ge N\) as defined upon \(F(x)\), and otherwise equals \[F_{[N]}(x) = \frac{F(x) - \lim_{\varepsilon \to 0^{+}} F(-N-\varepsilon)}{1-e_F(N)}.\] 

Assume instead that \(F(x)\) is a biased function, allowing \(0 < p < 1\) and allow \(\gamma^{2}_{M}\) to be any positive number, representing the amount by which we wish to shift the expectation of the distribution of \(X\) upon the expected asymptotic population defined by the Kemeny metric. Should  \(X\) be any random variable with a finite expectation, then allow \(\varepsilon > 0\). Pick an \(N\) for which \(e_{F}(N) \le \tfrac{\varepsilon}{2}\) and truncate \(X\) at \(N\). It then follows that the mean of \(E(X_{[N]}) = 0 + \gamma^{2}_{M}\) expressed as a \(\frac{\varepsilon}{2}\)-mixture, which changes the total variation distance by at most \(\tfrac{\varepsilon}{2}\). By sub-additivity then follows
\begin{equation}
\|X - Y\|_{TV} \le \|X - X_{[N]}\|_{TV} + \|X_{[N]} - Y\|_{TV} \le \frac{\varepsilon}{2} + \frac{\varepsilon}{2} = \varepsilon,
\end{equation}
as there exists no element \(N \notin U(\mathcal{M})\).

This therefore proves that for any \(X\) exists a finite expectation, and therefore and independently distributed random variable within finite expectation upon the Kemeny metric, there is always a way to truncate \(X\) and shift its mean to \(\gamma^{2}_{M}\), no matter what value \(\gamma^{2}_{M}\) might have, without moving by more than \(\varepsilon\) in the total variation distance. These constructions put an upper bound on \(\|Y\|\): it is no greater than the larger of $|\tfrac{N}{2}|$, representing the absolute value of the position of the atom located at \(2\tfrac{(\gamma^{2}_{M} - E(X))}{\varepsilon}\). In consequence,  the tails of $Y$ are zero, making them sub-Gaussian. As $\varepsilon$ may be arbitrarily small, the only possible lower bound on the distance is zero, and the affine linear invariance of any compact and totally bounded, or complete metric, space allows the estimated parameter to almost surely be upon \(U(\mathcal{M})\), for all finite \(n\).
\end{proof}

\begin{lemma}
\label{lem:upper_equality}
The total variational distance between any Euclidean linear function and an unbiased Kemeny linear distance function is upper-bounded by a finite function of \(\Pr_{\kappa}(X)\), corresponding to a finite distance of \(|\tfrac{n^{2}-n}{2}|\ge{0}.\)
\end{lemma}
\begin{proof}
The lower bound is always finite and may be treated as 0 w.l.g. for any affine linear function upon a metric space, for an arbitrary collection of points. The upper bound for the performance of a system \(\rho(\hat{Y},Y)\), for arbitrary metric \(\rho\) which may be indeterminate (due to the lack of a compact and totally bounded domain upon the extended real line). For the Euclidean metric space, this issue is identified by the use of `approximately correct systems' which bound the measure of the space to be a finite value. With the extended real line \(\overline{\mathbb{R}}\) of performance, for which in conjunction with the Kemeny metric \(\rho_{\kappa}\) we obtain finite moments for arbitrary measure spaces, we show that the total variational distance is almost surely finitely upper-bounded for any homogeneous function space (i.e., an affine linear function space for a common population). 

Allow \(X\) to not contain finite expectations and also be a random variable. Then as \(X\) is unmeasurable, for the moments are not in the real line upon the Euclidean metric space, the probability bounds are almost surely only guaranteed (measurable with probability 1) upon the Kemeny metric space for which the Kemeny distance support is defined \(U(\mathcal{M})\) about 0. For any non-constant vector \(X\) then, all measurable distances between \(X\) and \(Y\) are upon \(U(\mathcal{M})\), with realised error \(\sup |\tfrac{n^{2}-n}{2}|.\) Then the maximum error is almost surely in the neighbourhood about 0, \(\Pr(2|(\tfrac{n^{2}-n}{2})| \in U(\mathcal{M})) = 1\), guaranteeing convergence.
\end{proof}

\begin{corollary}~\label{cor:lyapunov}
Suppose \(\{X_{1},\ldots ,X_{n}\}\) is a sequence of independent random variables, each with finite expected value $\mu_{i}$ and variance $\sigma _{i}^{2}$. Define $s_{n}^{2}=\sum _{i=1}^{n}\sigma _{i}^{2}$; then for some $\delta>0$, Lyapunov's condition is satisfied, then a sum of \(\frac {X_{i}-\mu _{i}}{s_{n}}\) converges in distribution to a standard normal random variable, as $n \to \infty^{+}$:
\[\lim _{n\to \infty }\;{\frac {1}{s_{n}^{2+\delta }}}\,\sum _{i=1}^{n}\mathbb {E} \left[\left|X_{i}-\mu _{i}\right|^{2+\delta }\right]=0.\]
\end{corollary}
\begin{proof}
A simple verification of the Lyapunov condition is seen to hold for any finite $n$. This is because the skewness of any distribution of Kemeny distances, $\gamma^{3}_{\kappa}(\mathcal{M}_{n}) = 0,~\forall n = 1,\dots,\mathbb{N}^{+} < \infty^{+}$, by the even function property of the $\kappa$ function and any affine linear transformations thereof upon any Hilbert space (Definition~\ref{def:hilbert}; Lemma~\ref{lem:even} \& Lemma~\ref{lem:kem_expansion}). 

As \(X_1, X_2,\cdots X_n \in \mathcal{M}\) are i.i.d, $\sum_{i=1}^{\mathcal{M}} E[|X_{i}|^3] = n E(|X_n|^3)$. By Lyapunov's inequality, $\infty > E(|X_n|^3)\geq [E(|X_n|^2)]^{\frac{3}{2}}\geq (\sigma_{n}^{2})^{\frac{3}{2}} = \sigma_{n}^{3}$, which is true by the finite and totally compact nature of the even function. As $s_n^3=(\sum_{i=1}^n \sigma_i^2)^{\tfrac{3}{2}}=n^{\tfrac{3}{2}}\sigma^3$, thus follows $\dfrac{nE|X_1|^3}{\sqrt[1.5]{n}\sigma_{\kappa}^{3}}=C\dfrac{1}{\sqrt{n}}\to 0$.
% Lyapunov CLT[6] — Suppose { X 1 , … , X n , … } {\textstyle \{X_{1},\ldots ,X_{n},\ldots \}} {\textstyle \{X_{1},\ldots ,X_{n},\ldots \}} is a sequence of independent random variables, each with finite expected value μ i {\textstyle \mu _{i}} {\textstyle \mu _{i}} and variance σ i 2 {\textstyle \sigma _{i}^{2}} {\textstyle \sigma _{i}^{2}}. Define
% s n 2 = ∑ i = 1 n σ i 2 . {\displaystyle s_{n}^{2}=\sum _{i=1}^{n}\sigma _{i}^{2}.}
% {\displaystyle s_{n}^{2}=\sum _{i=1}^{n}\sigma _{i}^{2}.}

% If for some δ > 0 {\textstyle \delta >0} {\textstyle \delta >0}, Lyapunov’s condition
% lim n → ∞ 1 s n 2 + δ ∑ i = 1 n E [ | X i − μ i | 2 + δ ] = 0 {\displaystyle \lim _{n\to \infty }\;{\frac {1}{s_{n}^{2+\delta }}}\,\sum _{i=1}^{n}\mathbb {E} \left[\left|X_{i}-\mu _{i}\right|^{2+\delta }\right]=0}
% {\displaystyle \lim _{n\to \infty }\;{\frac {1}{s_{n}^{2+\delta }}}\,\sum _{i=1}^{n}\mathbb {E} \left[\left|X_{i}-\mu _{i}\right|^{2+\delta }\right]=0}
% is satisfied, then a sum of X i − μ i s n {\textstyle {\frac {X_{i}-\mu _{i}}{s_{n}}}} {\textstyle {\frac {X_{i}-\mu _{i}}{s_{n}}}} converges in distribution to a standard normal random variable, as n {\textstyle n} {\textstyle n} goes to infinity:
% {\displaystyle {\frac {1}{s_{n}}}\,\sum _{i=1}^{n}\left(X_{i}-\mu _{i}\right)\ \xrightarrow {d} \ {\mathcal {N}}(0,1).}
% {\displaystyle {\frac {1}{s_{n}}}\,\sum _{i=1}^{n}\left(X_{i}-\mu _{i}\right)\ \xrightarrow {d} \ {\mathcal {N}}(0,1).}

For any positive integer $n$, the distribution of the function is concluded symmetric, and from this it follows that the expectation of the boundary condition of the moments greater than 2, $2 + \inf{\delta} = 2 + 1 \implies \gamma^{(2 + \inf{\delta})}_{\kappa} = 0$ by the evenness of the $\kappa$ function upon $n$ reals (Lemma~\ref{lem:even}), the central limit theorem is therefore valid for any sequence of $n$ independently distributed realisations. 
 \end{proof}
\begin{corollary}~\label{cor:kolmogorov}
Kolmogorov's law applied to the Kemeny metric holds that the sample average converges almost surely to the expected value over the field of permutations \(\mathcal{M}\): \[\lim_{m\to\infty^{+}}{\overline {x}}_{m}\ \xrightarrow {\text{a.s.}} \mu \equiv \Pr \!\left(\lim _{m\to \infty^{+} }E(x_{m})_{\kappa} =\mu_{\kappa} \right)=1.\]
\end{corollary}
\begin{proof}
The strong law applies to independent identically distributed random variables having an expected value in orderability; if the variables are independent and identically distributed (as upon the Kemeny metric), then it is necessary that they have an expected value; Lemma~\ref{lem:sigma_upper}), and if the summands are independent but not identically distributed with a finite second moment (Lemma~\ref{lem:sigma_lower} \& Lemma~\ref{lem:sigma_upper}), follows:

    \[\lim_{m\to\infty^{+}}{\overline{x}}_{m} - E{\big [}{\overline {x}}_{m}{\big ]}\ {\xrightarrow {\text{a.s.}}}\ 0,\]
and therefore the Kolmogorov's law of convergence upon the Kemeny metric space validly holds for any indepedently distributed samples upon a population of \(\mathcal{M}\) permutations. 
\end{proof}

\section{Density of the permutation space \(\mathcal{M}\)}
\label{lem:density}
% \begin{multicols}{2}
 \begin{equation*}
\noindent
 \scriptsize
 \begin{split}
 n^{n} & < {\sqrt {2\pi n}}\ \left({\frac {n}{e}}\right)^{n}e^{\frac {1}{12n+1}}\\
n\log(n) & < \frac{1}{2} \log(2\pi n) + n (\log(n) - \log(e)) + \big(\log(1) - \log(12n +1)\big)\log(e) < n! \\
& \hspace{1cm} <  n\log(n)  < \frac{1}{2} \log(2\pi n) + n (\log(n) - \log(e)) + \big(\log(1) - \log(12n)\big)\log(e)\\
n\log(n) & < \frac{1}{2} \log(2\pi n) + n (\log(n) - 1) + \big(0 - \log(12n +1)\big) < n! \\
& \hspace{1cm}<  \frac{1}{2} \log(2\pi n) + n (\log(n) - 1) + \big(0 - \log(12n)\big)\\
n\log(n) & < \frac{1}{2} \log(2\pi) + n ( \log(n) + \log(n) - 1) \equiv (\log(n)) + \big(0 - \log(12n +1)\big)\\
 n^{n} & < \sqrt{2\pi n} \cdot (\frac{n}{e})^{n}\\
 n\log(n) & < \log(2\pi n) + n\big(\log(n) - \log(e)\big)\\
 n \log(n) & < \frac{1}{2} \log\big(2\pi n\big) + n \log(n) - n\log(e)\\
 0 & < \frac{1}{2}\log\big(2\pi n) - n\\
 2n & < \log\big(2\pi n\big)\\
 \end{split}
 \end{equation*}
 % \end{multicols}

\end{document}